\documentclass[12pt,draftcls, onecolumn,journal]{IEEEtran}
\usepackage{times}
\usepackage[reqno]{amsmath}
\usepackage{amsfonts}
\usepackage{caption3}
\usepackage{times,amsmath,epsfig}
\usepackage{latexsym,amssymb}
\usepackage{rotating,color}
\usepackage{cite}
\usepackage{extarrows}
\usepackage{amsthm}
\usepackage{algorithm}
\usepackage{algpseudocode}

\def\myQED{\mbox{\rule[0pt]{1.5ex}{1.5ex}}}

\newcommand{\no}{\nonumber}

\newtheorem{thm}{Theorem}[section]

\newtheorem{lem}[thm]{Lemma}
\newtheorem{cor}[thm]{Corollary}
\newtheorem{define}[thm]{Definition}
\newtheorem{rmk}[thm]{Remark}
\newtheorem{exmpl}[thm]{Example}
\IEEEoverridecommandlockouts

\makeatletter
\def\ps@headings{%
\def\@oddhead{\mbox{}\scriptsize\rightmark \hfil \thepage}%
\def\@evenhead{\scriptsize\thepage \hfil \leftmark\mbox{}}%
\def\@oddfoot{}%
\def\@evenfoot{}}
\makeatother
\begin{document}

\title{Multi-Chart Detection Procedure for Bayesian Quickest Change-Point Detection with Unknown Post-Change Parameters}

\author{Jun Geng, Erhan Bayraktar, Lifeng Lai\thanks{The work of J. Geng is supported by the National Natural Science Foundation of China under grant 61601144 and by the Fundamental Research Funds for the Central Universities under grant AUGA5710013915.
The work of E. Bayraktar is supported in part by the NSF under grant number DMS-1613170 and in part by the Susan M. Smith Professorship. The work of L. Lai is supported by the National Science Foundation under grants CNS-1660128 and ECCS-1711468. This paper was presented in part at IEEE International Conference on Acoustics, Speech, and Signal Processing (ICASSP), Shanghai, China, Mar. 2016~\cite{Jun:ICASSP:16_2}.

J. Geng is with the School of Electronics and Information Engineering, Harbin Institute of Technology, Harbin, 150001, China (Email: jgeng@hit.edu.cn). E. Bayraktar is with the Department of Mathematics, University of Michigan, Ann Arbor, MI 48109, USA (Email:erhan@umich.edu). L. Lai is with the Department of Electrical and Computer Engineering, University of California, Davis, CA, 95616, USA (Email: lflai@ucdavis.edu).}}
\maketitle 


\maketitle
\begin{abstract}
In this paper, the problem of quickly detecting an abrupt change on a stochastic process under Bayesian framework is considered. Different from the classic Bayesian quickest change-point detection problem, this paper considers the case where there is uncertainty about the post-change distribution. Specifically, the observer only knows that the post-change distribution belongs to a parametric distribution family but he does not know the true value of the post-change parameter. In this scenario, we propose two multi-chart detection procedures, termed as M-SR procedure and modified M-SR procedure respectively, and show that these two procedures are asymptotically optimal when the post-change parameter belongs to a finite set and are asymptotically $\epsilon-$optimal when the post-change parameter belongs to a compact set with finite measure. Both algorithms can be calculated efficiently as their detection statistics can be updated recursively. We then extend the study to consider the multi-source monitoring problem with unknown post-change parameters. When those monitored sources are mutually independent, we propose a window-based modified M-SR detection procedure and show that the proposed detection method is first-order asymptotically optimal when post-change parameters belong to finite sets. We show that both computation and space complexities of the proposed algorithm increase only linearly with respect to the number of sources.
\end{abstract}


\section{Introduction} \label{sec:intro}
To quickly detect the abrupt or the abnormal change in the observing sequence is of interest in wide range of practical applications. 
For example, in the cognitive radio system, the secondary user wants to quickly identify the time instant when the primary user accesses or releases the channel to maximize its throughput \cite{Kim:TSP:10, Sayed:TSP:10, Lai:GLOBE:08}. As another example, in seismic monitoring, it is crucial to quickly detect the abnormal signal caused by the earth crust movement. In such applications, to minimize the detection delay, which is the difference between the time when the abnormal change occurs and the time when the change is declared, is of interest. Quickest change-point detection (QCD) is a suitable mathematical framework for such applications. In particular, QCD aims to design online algorithms that can identify the abrupt change in the probabilistic distribution of a stochastic process as quickly and accurately as possible. 
QCD has two main classes of problem formulations: Bayesian formulation \cite{Shiryaev:Soviet:61, Shiryaev:TPIA:63} and non-Baysian formulation \cite{Lorden:AmS:71,Pollak:AnS:85}. For the classic Bayesian formulation, one sequentially observes a stochastic process $\{X_{k}\}_{k=1}^{\infty}$ with a random change-point $t$. Before the change-point $t$, the sequence $X_{1}, \ldots, X_{t-1}$ are independent and identically distributed (i.i.d.) with probability density function (pdf) $f_{0}$, and after $t$, the sequence $X_{t}, X_{t+1}, \ldots$ are i.i.d. with pdf $f_{1}$. In the Bayesian formulation, the change-point $t$ is typically modeled as a geometrically distributed random variable. The goal is to find an optimal stopping time $\tau$, at which we declare the change has happened, that minimizes the average detection delay $\mathbb{E}[(\tau - t)^{+}]$ subjected to a false alarm constraint. The Shiryaev-Robert (SR) procedure is known to be the optimal detection procedure for Bayesian QCD \cite{Shiryaev:TPIA:63, Tartakovsky:TPIA:04}. In the non-Bayesian formulation, the change-point $t$ is assumed to be a fixed but unknown constant.

In the classic Bayesian QCD, it assumes that both the pre-change distribution $f_{0}$ and the post-change distribution $f_{1}$ are perfectly known by the observer. In most of practical applications, it is reasonable to assume that the pre-change distribution is known, as we typically know the behavior of the system in the normal state. However, one may not know the post-change distribution perfectly, as this represents the system in the abnormal states. Motivated by this fact, in this paper we extend the classic Bayesian QCD problem to the case with incomplete post-change information. In particular, we focus on the case that the post-change distribution belongs to a parametric distribution family, but the true post-change parameter, which belongs to a known compact set $\Lambda$, is unknown to the observer. The goal of the observer is to design an effective online algorithm to quickly detect the change-point in the observation sequence for \emph{all} possible post-change parameter in $\Lambda$. In this scenario, we propose two low complexity multi-chart procedures for the purpose of quickest detection. In the first proposed detection algorithm, the observer divides $\Lambda$ into $I$ small disjoint subsets. By selecting one candidate point in each small subset, the observer constructs a finite set $\Lambda_{D} \subseteq \Lambda$. Then, for each elements within $\Lambda_{D}$, the observer runs a SR detection procedure; the observer declares that the change has occurred when any one of these procedures stops. We note that the observer runs multiple SR procedures simultaneously in each time slot, we term this algorithm as M-SR procedure. The second proposed algorithm is similar to the first one except that it replaces the summation operator within the SR statistic by the maximum operator; hence we term the second proposed multi-chart procedure as the modified M-SR procedure. For both proposed algorithms, we show that they are asymptotically optimal for the post-change parameters within $\Lambda_{D}$ and asymptotically $\epsilon-$optimal for all parameters within $\Lambda$ if $\Lambda_{D}$ is properly selected. The definition of asymptotic $\epsilon-$optimality will be given explicitly in the sequel. Loosely speaking, it indicates that the performance loss of the proposed algorithm is no larger than a small constant $\epsilon$ as the false alarm probability vanishes. We further point out that both proposed algorithms can be calculated recursively hence they are computationally efficient.

We then further extend the study to the multi-source monitoring problem under Bayesian QCD framework. In particular, the observer monitors $L$ mutually independent sources (except the change happens at the same time) and aims to detect the geometrically distributed change-point $t$ as quickly as possible. We consider the case that the distributions of all $L$ sources change simultaneously at $t$ but the post-change distribution of each source contains one unknown parameter. If we directly apply the M-SR or the modified M-SR procedure to the case of multi-variate unknown parameter, the computational complexity will increase exponentially w.r.t. $L$. In this scenario, we propose a window-based modified M-SR procedure. We analyze the complexity and the performance of the proposed algorithm in detail, and show that the computational complexity of this proposed algorithm increases only linearly w.r.t. $L$ and the proposed algorithm is asymptotically optimal as the  false alarm probability goes to zero.

The problem considered in this paper is related to recent works on the QCD problem that take the unknown post-change parameter into consideration. Due to limited space, we mention a few of them. For Bayesian setups, optimal solutions are likely to be obtained by converting proposed problems into Markovian optimal stopping problems. 
For example, \cite{Bayraktar:AAP:06} solves the Poisson disorder problem with unknown post-change arrival rate; \cite{Dayanik:MOR:08} solves the Bayesian sequential change diagnosis problem, in which the observer aims not only to quickly detect the change but also to accurately identify the post-change parameter from a finite set; \cite{Dayanik:TIT:09} solves a generalized formulation of the sequential change diagnosis problem for a Markov modulated sequence. For non-Bayeisan QCD setups, most of existing works propose modified versions of the CUSUM detection procedure, which is the optimal scheme for the classic non-Bayeisan QCD problem with known parameters \cite{Moustakides:AnS:86}, and discuss their robustness over the post-change uncertainty. For example, \cite{Lorden:AmS:71, Lorden:AnS:73, Lai:JRSS:95} replace the unknown likelihood ratio in the CUSUM statistic by the generalized likelihood ratio (GLR). 
\cite{Pollak:AnS:87} proposes the mixture-based CUSUM algorithm in which the CUSUM statistic is averaged over a prior distribution of the unknown parameter.  \cite{Lai:TIT:98} further shows the asymptotic optimality of the GLR-based CUSUM and the mixture-based CUSUM detection algorithms. \cite{Yang:Stochastics:17} designs a composite stopping time that combines multiple CUSUM procedures to quickly detect the disorder time in the Wiener process with post-change uncertainty. \cite{Banerjee:TIT:15} mentions the multi-chart CUSUM detection strategy. One may refer to a recent book \cite{Tartakovsky:Book:14} and references therein for more detailed results of this topic.

Different from aforementioned literatures, our work is formulated under Bayesian framework and focuses on the performance of low complexity multi-chart detection procedures instead of optimal schemes, which usually have high computational complexities. \cite{Dayanik:AOR:13} has studied the asymptotic detection rules for Bayesian sequential diagnosis problem. Our paper focuses on the QCD problem and shows that the multi-chart version of the well-known SR procedure exhibits robustness over unknown parameters. The problem studied in this paper can be also viewed as a special case of the Markov chain tracking problem \cite{Bayraktar:SPA:09, Bayraktar:SADA:09, Bayraktar:AOR:10} 
with several absorbing states. 
However, those works focus on the optimal solution of generalized formulations while our work focus on the asymptotic optimal algorithms for the QCD problem.

There are also many recent works considering the QCD problem for multi-source monitoring. \cite{Tartakovsky:ICIF:08} is one of the most related works. Specifically, \cite{Tartakovsky:ICIF:08} considers the QCD problem in a distributed multi-sensor system when the post-change parameters are unknown. Authors also proposed to use multi-chart CUSUM to solve the proposed problem. In addition, \cite{Mei:Biometrika:10} proposes the SUM algorithm, which is based on the sum of local CUSUMs, to quickly detect the abrupt change in multiple independent data streams. However, \cite{Mei:Biometrika:10} only focuses on the case with known post-change distributions. Both \cite{Tartakovsky:ICIF:08} and \cite{Mei:Biometrika:10} are formulated under non-Bayesian setting. We consider the multi-source monitoring problem under Bayesian QCD framework and discuss the performance and the implementation complexity of the proposed algorithm in detail.

This paper extends our previous conference publication \cite{Jun:ICASSP:16_2} in several ways. Specifically, \cite{Jun:ICASSP:16_2} only discusses the asymptotic optimality of the M-SR procedure when the post-change parameter belongs to a finite set. In addition to including the contributions made in \cite{Jun:ICASSP:16_2}, this paper also proposes the modified M-SR procedure, and studies the performance of both algorithms under a more general setting. Furthermore, we consider the multi-source monitoring problem and analyze the performance of our newly proposed algorithm.

The remainder of this paper is organized as follows. The mathematic model is given in Section \ref{sec:model}. Section \ref{sec:optimal} presents the proposed multi-chart detection algorithms and analyzes their asymptotic performances. Section \ref{sec:ext} discusses the multi-source monitoring problem and analyzes the performance of the proposed window-based modified M-SR procedure. Technical proofs in this paper are presented in Section \ref{sec:proof}. Numerical examples are given in Section \ref{sec:sim} to illustrate the theoretic results obtained in this work. Finally, Section \ref{sec:con} offers concluding remarks.


\section{Model} \label{sec:model}
\subsection{Problem Formulation}
We consider a random observation sequence $\left\{ X_{k}, k = 1, 2, \ldots \right\}$ whose distribution changes at an unknown time $t$. In particular, observations obtained before the change-point $t$, namely $X_{1}, X_{2}, \ldots, X_{t-1}$, are i.i.d. with pdf $g(x; \theta_{0})$; observations obtained after the change-point $t$, namely $X_{t}, X_{t+1}, \ldots$, are i.i.d. with pdf $f(x; \lambda)$.

In the classic QCD problem, both the pre-change distribution $g(x; \theta_{0})$ and the post-change distribution $f(x; \lambda)$ are perfectly known by the observer. In this paper, we consider the case that the observer only knows partial information of the post-change distribution. Specifically, we assume that the pre-change distribution $g(x; \theta_{0})$ is perfectly known to the observer; hence $\theta_{0}$ is a known parameter. However, the post-change distribution $f(x; \lambda)$ contains an \emph{unknown} parameter $\lambda$. The observer knows that $\lambda$ is taken from a compact set $\Lambda$ but he does not know the true value of $\lambda$. 
For notation convenience, in the rest of paper, $g(x, \theta_{0})$ is also denoted as $g_{\theta_{0}}(x)$, and $f(x; \lambda)$ is denoted as $f_{\lambda}(x)$.

In this paper, we focus on the Bayesian QCD problem. In Bayesian framework, the change-point $t$ is modeled as a geometric random variable with known parameter $\rho$, i.e., for $0 < \rho < 1$,
\begin{eqnarray}
P(t = k) = \rho(1-\rho)^{k-1}, \quad  k = 1, 2, \ldots.  \label{eq:ext_geometry}
\end{eqnarray}
Observations $X_{k}$'s generate the filtration $\{\mathcal{F}_k\}_{k\in\mathbb{N}}$ with
$$\mathcal{F}_k=\sigma(X_1,\cdots,X_k), \quad k=1, 2, \ldots,$$
and $\mathcal{F}_{0}$ contains the sample space $\Omega$.

To facilitate the presentation, we denote $P_{k}^{(\theta_{0}, \lambda_{i})}$ as the conditional probability measure of the observation sequence given $\{t=k; \lambda=\lambda_{i}\}$.
For a measurable event $F$, we define probability measure $P_{\pi}^{(\theta_{0}, \lambda_{i})}$ as
\begin{eqnarray}
&&P_{\pi}^{(\theta_{0}, \lambda_{i})}(F) := \sum_{k=1}^{\infty} P_{k}^{(\theta_{0}, \lambda_{i})}(F)P(t=k). \no
\end{eqnarray}
We use $\mathbb{E}_{k}^{(\theta_{0}, \lambda_{i})}$ and $\mathbb{E}_{\pi}^{(\theta_{0}, \lambda_{i})}$ to denote the expectations with respect to probability measures $P_{k}^{(\theta_{0}, \lambda_{i})}$ and $P_{\pi}^{(\theta_{0}, \lambda_{i})}$, respectively.

The observer aims to detect the change-point as quickly and accurately as possible. Let $\mathcal{T}$ be a set of finite stopping times adapted to $\mathcal{F}_{k}$. A stopping time $\tau \in \mathcal{T}$ indicates the time instant that the observer stops taking observation and declares that the change has occurred. If the observer raises the alarm before the change-point occurs, i.e. $\tau < t$, we call the observer makes a false alarm. On the other hand, if $\tau \geq t$, we define $(\tau - t)$ as the detection delay. Hence, we use the following two performance metrics to evaluate a detection procedure $\tau$:
\begin{eqnarray}
&&\mathrm{ADD}(\tau; \theta_{0}, \lambda) := \mathbb{E}_{\pi}^{(\theta_{0}, \lambda)}[(\tau - t)^{+}], \no\\
&&\mathrm{PFA}(\tau; \theta_{0}, \lambda) := P_{\pi}^{(\theta_{0}, \lambda)}(\tau < t). \no
\end{eqnarray}
In our setup, the observer aims to minimize the average detection delay (ADD) while keeping the probability of false alarm (PFA) under control regardless of the value of $\lambda$. In other words, the observer wants to solve the following optimization problem:
\begin{eqnarray}
&& \inf_{\tau \in \mathcal{T}} \mathrm{ADD}(\tau; \theta_{0}, \lambda) \quad \text{ for all } \lambda \in \Lambda, \no\\
&& \textrm{subject to } \sup_{\lambda \in \Lambda} \mathrm{PFA}(\tau; \theta_{0}, \lambda) \leq \alpha, \label{eq:P1}
\end{eqnarray}
where $\alpha$ is a constant to control the false alarm. We note that \eqref{eq:P1} aims to minimize ADD and to achieve a small probability of false alarm \emph{for all possible post-change parameters simultaneously}. In general, it is difficult to find an optimal solution for the above multi-objective optimization problem. In this paper, we aim to design asymptotically optimal or sub-optimal detection algorithms.

\subsection{Other Related Formulations}
We emphasize that, in our formulation stated in \eqref{eq:P1}, the observer requires no prior distribution of $\lambda \in \Lambda$. In some existing related literatures, such as \cite{Bayraktar:AAP:06, Yang:Stochastics:17, Dayanik:AOR:13}, the prior distribution of the post-change parameter is assumed to be known by the observer. When the prior distribution of $\Lambda$ is available, in our context, it is natural to consider the problem formulation with averaged or weighted performance metrics.

Specifically, let $G(\lambda)$ be the prior distribution (or a subjective normalized weighting function) of $\Lambda$, and
\begin{eqnarray}
\int_{\lambda\in\Lambda} dG(\lambda) = 1. \label{eq:weighting}
\end{eqnarray}
In this case, we can define an averaged PFA as
\begin{eqnarray}
P_{\pi, \omega}^{(\theta_{0})}(\tau < t) 
:= \int_{\lambda \in \Lambda} P_{\pi}^{(\theta_{0}, \lambda)}(\tau < t) dG(\lambda), \label{eq:avg_pfa}
\end{eqnarray}
and define an averaged ADD as
\begin{eqnarray}
\mathbb{E}_{\pi, \omega}^{(\theta_{0})}[(\tau-t)^{+}] 
:= \int_{\lambda \in \Lambda} \mathbb{E}_{\pi}^{(\theta_{0}, \lambda)}[(\tau-t)^{+}] dG(\lambda). \label{eq:avg_delay}
\end{eqnarray}
Correspondingly, the observer may want to solve the following optimization problem
\begin{eqnarray}
&& \inf_{\tau \in \mathcal{T}} \mathbb{E}_{\pi, \omega}^{(\theta_{0})}[(\tau-t)^{+}], \no\\
&& \textrm{subject to } P_{\pi, \omega}^{(\theta_{0})}(\tau < t) \leq \alpha. \label{eq:PP1}
\end{eqnarray}

We note that the (asymptotically) optimal solution of Formulation \eqref{eq:P1} is also (asymptotically) optimal for Formulation \eqref{eq:PP1}. To see this fact, it is worth to notice that
$$ P_{\pi}^{(\theta_{0}, \lambda_{i})}(\tau < t) = P_{\pi}^{(\theta_{0}, \lambda_{j})}(\tau < t) \quad \forall \lambda_{i}, \lambda_{j} \in \Lambda, \lambda_{i}\neq\lambda_{j}.$$
This is because all observations are generated from $g_{\theta_0}(x)$ on the event $\{ \tau < t\}$ (when false alarm happens); hence the false alarm probability is independent of the post-change parameter $\lambda$. As a result, we have
\begin{eqnarray}
P_{\pi, \omega}^{(\theta_{0})}(\tau < t) = P_{\pi}^{(\theta_{0}, \lambda_{i})}(\tau < t) = \sup_{\lambda} P_{\pi}^{(\theta_{0}, \lambda)}(\tau < t). \label{eq:avg_worst}
\end{eqnarray}
Hence, if the worst case false alarm constraint in \eqref{eq:P1} is satisfied by detection procedure $\tau$, the average false alarm constraint in \eqref{eq:PP1} is also satisfied. Furthermore, for the objective function, we note that to minimize $\textrm{ADD}(\tau; \theta_{0}, \lambda)$ for all $\lambda \in \Lambda$ simultaneously in \eqref{eq:P1} is stronger than to minimize $\mathbb{E}_{\pi, \omega}^{(\theta_{0})}[(\tau-t)^{+}]$, the detection delay averaged over $\lambda$, in \eqref{eq:avg_delay}. Hence, Formulation \eqref{eq:P1} is more stringent than Formulation \eqref{eq:PP1}. As a result, the (asymptotically) optimal solution for Formulation \eqref{eq:P1} is also expected to have a good performance for Formulation \eqref{eq:PP1}.

\section{Multi-Chart Procedures and Their Sub-Optimality} \label{sec:optimal}
It is well known that the SR procedure is the optimal detection procedure for classic Bayesian QCD problem. In this section, we propose two new detection algorithms, both of which are modified versions of the SR procedure, and analyze their performances.

\subsection{Multi-chart procedures}
For the classic Bayesian QCD problem with pre-change distribution $f_{0}$ and post-change distribution $f_{1}$, the SR procedure is known as
\begin{eqnarray}
&&R_{n} := \sum_{k=1}^{n} \prod_{q=k}^{n} \frac{1}{1-\rho}\frac{f_{1}(X_{q})}{f_{{0}}(X_{q})}, \no\\
&&\tau_{SR} := \inf\left\{ n \geq 1 :  \log R_{n} >  \log B \right\}, \no
\end{eqnarray}
for a properly chosen threshold $B$. For our formation, as the post-change distribution contains an unknown parameter, it is natural to consider replacing the likelihood ratio in the SR procedure with the generalized likelihood ratio. In our context, the GLR-based SR procedure can be written as
\begin{eqnarray}
&&R_{n}^{*} := \sup_{\lambda \in \Lambda} \sum_{k=1}^{n} \prod_{q=k}^{n} \frac{1}{1-\rho}\frac{f_{\lambda}(X_{q})}{g_{\theta_{0}}(X_{q})}, \no\\
&&\tau_{GLR} := \inf\left\{ n \geq 1 :  \log R_{n}^{*} >  \log B \right\}. \no
\end{eqnarray}
Though the GLR-based SR procedure is natural and attractive, it has several shortcomings. Theoretically, it is a challenging task to find the false alarm probability of $\tau_{GLR}$. Practically, $\tau_{GLR}$ has large computational complexity as the observer needs to solve for the best $\lambda^{*}$ at each time slot $n$. The computational complexity depends on the specific form of $g_{\theta_{0}}$ and $f_{\lambda}$, which could lead to very challenging optimization problems.

In order to resolve these difficulties, we discretize the feasible set $\Lambda$ by selecting finitely many 
candidate points. Specifically, let $\Lambda_{D} = \{ \lambda_{0}, \lambda_{1}, \ldots, \lambda_{I-1}\}$ be a discrete set with $I$ different elements, and $\Lambda_{D} \subseteq \Lambda$. For $i=0, \ldots, I-1$, let
\begin{eqnarray}
R_{n}^{(\theta_{0}, \lambda_{i})} := \sum_{k=1}^{n} \prod_{q=k}^{n} \frac{1}{1-\rho}\frac{f_{\lambda_{i}}(X_{q})}{g_{\theta_{0}}(X_{q})} \label{eq:stat_R}
\end{eqnarray}
be the detection statistic, we propose the following detection procedure:
\begin{eqnarray}
&&\tau_{R}^{(i)} = \inf\left\{ n \geq 1 : \log R_{n}^{(\theta_{0}, \lambda_{i})} > \log B_{i} \right\}, \label{eq:tau_ri}\\
&&\tau_{R} = \min_{i\in\{0, \ldots, I-1\}} \tau_{R}^{(i)}. \label{eq:tau_r}
\end{eqnarray}
There are a few comments about the proposed algorithm. Firstly, the proposed algorithm can be viewed as a GLR-based SR detection procedure over the discrete set $\Lambda_{D}$. As we use $\Lambda_{D}$ to approximate the post-change parameter set, $\tau_{R}$ then can be viewed as an approximation of $\tau_{GLR}$. Secondly, the proposed algorithm is a multi-chart detection procedure. Specifically, in the proposed algorithm, the observer updates $I$ statistics $R_{n}^{(\theta_{0}, \lambda_{i})}$, $i=0, \ldots, I-1$, in a parallel manner at each time slot, and each statistic is compared with its own threshold $B_{i}$. The procedure stops when any one of statistics exceeds its corresponding threshold. As $R_{n}^{(\theta_{0}, \lambda_{i})}$ is the statistic used in the SR procedure, we term our proposed multi-chart procedure as \emph{M-SR procedure}. Thirdly, we comment that $\tau_{R}$ can be calculated efficiently since $R_{n}^{(\theta_{0}, \lambda_{i})}$ can be updated recursively at each time slot. It is easy to see
\begin{eqnarray}
R_{n}^{(\theta_{0}, \lambda_{i})} = \left(1+R_{n-1}^{(\theta_{0}, \lambda_{i})}\right)\frac{1}{1-\rho}\frac{f_{\lambda_{i}}(X_{n})}{g_{\theta_{0}}(X_{n})}. \label{eq:R_recuresive}
\end{eqnarray}
Hence, the computational complexity of the proposed M-SR procedure is on the order of $O(I)$ at each time slot.

Besides the M-SR procedure, we also propose the following \emph{modified M-SR procedure}:
\begin{eqnarray}
&&C_{n}^{(\theta_{0}, \lambda_{i})} := \max_{1\leq k \leq n} \prod_{q=k}^{n} \frac{1}{1-\rho}\frac{f_{\lambda_{i}}(X_{q})}{g_{\theta_{0}}(X_{q})}, \quad i = 0, \ldots, I-1,  \label{eq:stat_C}\\
&&\tau_{C}^{(i)} = \inf\left\{ n \geq 1 : \log C_{n}^{(\theta_{0}, \lambda_{i})} > \log B_{i} \right\}, \label{eq:tau_ci}\\
&&\tau_{C} = \min_{i\in\{0, \ldots, I-1\}} \tau_{C}^{(i)}. \label{eq:tau_c}
\end{eqnarray}
We note that the difference between $\tau_{R}$ and $\tau_{C}$ is that $R_{n}^{(\theta_{0}, \lambda_{i})}$ takes summation over $1$ to $n$ but $C_{n}^{(\theta_{0}, \lambda_{i})}$ takes the maximum. It is worth to notice that $C_{n}^{(\theta_{0}, \lambda_{i})}$ can also be updated recursively as
\begin{eqnarray}
C_{n}^{(\theta_{0}, \lambda_{i})} = \max\left\{C_{n-1}^{(\theta_{0}, \lambda_{i})}, 1 \right\}\frac{1}{1-\rho}\frac{f_{\lambda_{i}}(X_{n})}{g_{\theta_{0}}(X_{n})}.
\end{eqnarray}
Hence, the computational complexity of the proposed modified M-SR procedure is also on the order of $O(I)$ at each time slot.

In the rest of this section, we will analyze the performance of $\tau_{R}$ and $\tau_{C}$ for our proposed Formulation \eqref{eq:P1} under different assumptions on $\Lambda$. Before our further analysis, we first present asymptotic lower bounds of ADDs for the  post-change parameters in $\Lambda$ in the following theorem:
\begin{thm} \label{lem:LB}
(Lower Bounds) 
For all $\lambda_{i} \in \Lambda$, as $\alpha \rightarrow 0$,
\begin{eqnarray}
&& \inf_{\tau \in \mathcal{T} }\left\{ \mathbb{E}_{\pi}^{(\theta_{0}, \lambda_{i})}[(\tau - t)^{+}]: \sup_{\lambda \in \Lambda_{D}} P_{\pi}^{(\theta_{0}, \lambda)}(\tau < t) \leq \alpha \right\} \no\\
&&\hspace{10mm} \geq \frac{|\log \alpha|}{D(f_{\lambda_{i}}, g_{\theta_{0}})+|\log(1-\rho)|}(1+o(1)),  \label{eq:LB_P1}
\end{eqnarray}
where $D(f_{\lambda_{i}}, g_{\theta_{0}})$ is the Kullback-Leibler (KL) divergence between $f(x; \lambda_{i})$ and $g(x; \theta_{0})$.
\end{thm}
\begin{proof}
\eqref{eq:LB_P1} can be shown as follows:
\begin{eqnarray}
&&\inf_{\tau \in \mathcal{T}}\left\{ \mathbb{E}_{\pi}^{(\theta_{0}, \lambda_{i})}[(\tau - t)^{+}]: \sup_{\lambda \in \Lambda} P_{\pi}^{(\theta_{0}, \lambda)}(\tau < t) \leq \alpha \right\} \no\\
&&\hspace{10mm} \geq \inf_{\tau \in \mathcal{T} }\left\{ \mathbb{E}_{\pi}^{(\theta_{0}, \lambda_{i})}[(\tau - t)^{+}]: P_{\pi}^{(\theta_{0}, \lambda_{i})}(\tau < t) \leq \alpha \right\} \no\\
&&\hspace{10mm} \geq \frac{|\log \alpha|}{D(f_{\lambda_{i}}, g_{\theta_{0}})+|\log(1-\rho)|}(1+o(1)). \no
\end{eqnarray}
The last inequality is from the lower bound of the average detection delay in the classic Bayesian QCD problem (see Theorem 1 in \cite{Tartakovsky:TPIA:04}).
\end{proof}

\subsection{M-SR procedure and posterior probabilities}
The statistics in the M-SR procedure can be equivalently expressed in terms of posterior probabilities. Let $\Lambda_{D} = \{\lambda_{0}, \ldots, \lambda_{I-1}\}$ be the discrete set for implementing the M-SR procedure. Define posterior probabilities as
\begin{eqnarray}
\pi_{n}^{(\theta_{0}, \lambda_{i})} := P_{\pi}^{(\theta_{0}, \lambda_{i})}(t \leq n | \mathcal{F}_{n}), \quad i = 0, \ldots, I-1. \label{eq:def_postp}
\end{eqnarray}
We emphasize that \eqref{eq:def_postp} is defined on a probability measure given $\{\lambda = \lambda_{i}\}$ being the true post change parameter. By the definition of posterior probability, we have 
\begin{eqnarray}
\pi_{n}^{(\theta_{0}, \lambda_{i})} = \frac{\varrho_{n}^{(i)}(X_{1}, \ldots, X_{n})}{o_{n}(X_{1}, \ldots, X_{n}) +\varrho_{n}^{(i)}(X_{1}, \ldots, X_{n})},
\end{eqnarray}
in which
\begin{eqnarray}
&&\hspace{-8mm}o_{n}(X_{1}, \ldots, X_{n}) = (1-\rho)^{n}\prod_{q=1}^{n}g_{\theta_{0}}(X_{q}), \no\\
&&\hspace{-8mm}\varrho_{n}^{(i)}(X_{1}, \ldots, X_{n}) = \rho\sum_{k=1}^{n}(1-\rho)^{k-1}\prod_{q=1}^{k-1}g_{\theta_{0}}(X_{q})\prod_{q=k}^{n}f_{\lambda_{i}}(X_{q}). \no
\end{eqnarray}
By direct calculations, it is easy to verify that
\begin{eqnarray}
\log \frac{\pi_{n}^{(\theta_{0}, \lambda_{i})}}{1-\pi_{n}^{(\theta_{0}, \lambda_{i})}} = \log \rho + \log R_{n}^{(\theta_{0}, \lambda_{i})}.  \label{eq:Lambda1}
\end{eqnarray}
\eqref{eq:Lambda1} reveals an insightful relationship between the M-SR procedure and the posterior probability of the change-point occurrence, it plays an important role in analyzing the false alarm probability of the proposed algorithm. Despite $\Lambda$ is an uncountable set in general, the proposed M-SR procedure only has finitely many statistics $R_{n}^{(\theta_{0}, \lambda_{i})}$. Since each $R_{n}^{(\theta_{0}, \lambda_{i})}$ corresponds to a posterior probability by $\eqref{eq:Lambda1}$, it is possible to bound the false alarm probability for each $\tau_{R}^{(i)}$. Then, the overall false alarm probability can be bounded by the union bound inequality. The detailed technical proof is presented in Lemma \ref{lem:Pfa} in the sequel.


\subsection{$\epsilon-$ optimality of the proposed multi-chart procedures} \label{sec:optimal_interval}

Let $\Lambda_{D} = \{\lambda_{0}, \ldots, \lambda_{I-1}\}$ with $\lambda_{0} < \lambda_{2} < \ldots < \lambda_{I-1} $ be the set for implementing the M-SR procedure and the modified M-SR procedure.
Since $\Lambda_{D}$ is a finite set in general but $\Lambda$ is an uncountable infinite set, $\tau_{R}$ and $\tau_{C}$ constructed from $\Lambda_{D}$ will not be asymptotically optimal for all points in $\Lambda$. To measure the loss of optimality, we define $\epsilon-$optimality as follows:
\begin{define}
Let $\tau_{0}$ be a detection procedure that satisfies the false alarm constraint. $\tau_{0}$ is called (asymptotically) $\epsilon-$optimal if
\begin{eqnarray}
\eta(\tau_{0}; \lambda) := \lim_{\alpha \rightarrow 0} \frac{\inf_{\tau \in \mathcal{T}}\left\{ \mathrm{ADD}(\tau; \theta_{0}, \lambda): \mathrm{PFA}(\tau; \theta_{0}, \lambda) \leq \alpha \right\}}{\mathrm{ADD}(\tau_{0}; \theta_{0}, \lambda)} \geq 1-\epsilon \label{eq:e_opt}
\end{eqnarray}
for all $\lambda \in \Lambda$.
\end{define}
The definition above resembles the idea of $\epsilon-$optimality for the case of non-Bayesian QCD problem, which is originally proposed in [\citen{Nikiforov:SP:01}, Section 2.3.3]. In particular, given the true post-change parameter being $\lambda$, the numerator of \eqref{eq:e_opt} is the infimum of ADD for all qualified detection procedure; hence it presents the lower bound of the detection delay. The denominator of \eqref{eq:e_opt} is the ADD of a given detection procedure $\tau_{0}$. As a result, $\eta(\tau_{0}; \lambda)$ is a measurement of detection efficiency or optimality loss; it is easy to see that $0 \leq \eta(\tau_{0}; \lambda) \leq 1$ and that $\tau_{0}$ is asymptotically optimal for post-change parameter $\lambda$ when $\eta(\tau_{0}; \lambda) = 1$. We note that for a given detection procedure $\tau_{0}$, the detection efficiency $\eta(\tau_{0}; \lambda)$ also depends on $\lambda$; hence we call $\tau_{0}$ asymptotically $\epsilon-$optimal if \eqref{eq:e_opt} holds for all $\lambda\in\Lambda$.

To study the asymptotic performance of $\tau_{R}$ and $\tau_{C}$, we need to impose some additional assumptions on $f_{\lambda_{i}}$ and $g_{\theta_{0}}$. Specifically, for any given $\varepsilon > 0$, we define the random variable
\begin{eqnarray}
T_{\lambda_{i}, \varepsilon}^{(k, \lambda)} := \sup\left\{ n \geq 1: \Bigg|\frac{1}{n}\sum_{q=k}^{k+n-1}\left[ \frac{f_{\lambda_{i}}(X_{q})}{g_{\theta_{0}}(X_{q})} + |\log(1-\rho)| \right] - d(\theta_{0}, \lambda_{i}; \lambda) \Bigg| > \varepsilon \right\}, \label{eq:Tk}
\end{eqnarray}
in which
\begin{eqnarray}
d(\theta_{0}, \lambda_{i}; \lambda) = D(f_{\lambda}, g_{\theta_{0}}) - D(f_{\lambda}, f_{\lambda_{i}}) + |\log(1-\rho)|,
\end{eqnarray}
and the supremum of an empty set is defined as $0$. We note that when $\lambda$ is the true post-change parameter, on the event $\{ t=k \}$, we have
\begin{eqnarray}
\frac{1}{n}\sum_{q=k}^{k+n-1}\left[ \frac{f_{\lambda_{i}}(X_{q})}{g_{\theta_{0}}(X_{q})} + |\log(1-\rho)| \right] \rightarrow d(\theta_{0}, \lambda_{i}; \lambda), \quad P_{k}^{(\theta_{0}, \lambda)}-\textrm{almost surely}.
\end{eqnarray}
Hence, we have $T_{\lambda_{i}, \varepsilon}^{(k, \lambda)} < \infty$ almost surely under $P_{k}^{(\theta_{0}, \lambda)}$ for all $\lambda \in \Lambda$. We make additional assumptions that for any given $\lambda_{i} \in \Lambda_{D}$, $T_{\lambda_{i}, \varepsilon}^{(k, \lambda)}$ satisfies the condition
\begin{eqnarray}
\mathbb{E}_{k}^{(\theta_{0}, \lambda)}\left[ T_{\lambda_{i}, \varepsilon}^{(k, \lambda)} \right] < \infty \quad \forall \varepsilon > 0, \quad \forall k \geq 1 \text{ and } \forall \lambda\in\Lambda \label{eq:quickly_converge},
\end{eqnarray}
and
\begin{eqnarray}
\mathbb{E}_{\pi}^{(\theta_{0}, \lambda)}\left[ T_{\lambda_{i}, \varepsilon}^{(\lambda)}\right] = \sum_{k=1}^{\infty} \mathbb{E}_{k}^{(\theta_{0}, \lambda)}\left[ T_{\lambda_{i}, \varepsilon}^{(k, \lambda)} \right]P(t = k) < \infty, \quad \forall \varepsilon > 0 \text{ and } \forall \lambda\in\Lambda. \label{eq:averege_quickly}
\end{eqnarray}
With these assumptions, we have the following conclusion for the M-SR procedure and the modified M-SR procedure:
\begin{thm}\label{thm:general_opt}
Let $\lambda$ be the true post-change parameter. As $\alpha \rightarrow 0$, $\tau_{R}$ defined in \eqref{eq:tau_r} and $\tau_{C}$ defined in \eqref{eq:tau_c} satisfy the false alarm constraint by setting $B_{0} = \ldots = B_{I-1} = I(\rho\alpha)^{-1}$. In addition, the detection delay
\begin{eqnarray}
\mathbb{E}_{\pi}^{(\theta_{0}, \lambda)}[(\tau_{R} - t)^{+}]  &\leq& \mathbb{E}_{\pi}^{(\theta_{0}, \lambda)}\left[(\tau_{C}-t)^{+}\right] \no\\
&\leq& \frac{|\log \alpha|}{D(f_{\lambda}, g_{\theta_{0}}) - \min_{\lambda_{i} \in \Lambda_{D}} D(f_{\lambda}, f_{\lambda_{i}}) + |\log(1-\rho)|}(1+o(1)). \no
\end{eqnarray}
\end{thm}
\begin{proof}
Please see Section \ref{subsec:optimal}.
\end{proof}
\begin{rmk}
\eqref{eq:quickly_converge} and \eqref{eq:averege_quickly} are modifications of the well known assumptions ``$r$-quick convergence'' and ``average-$r$-quick convergence''\cite{Tartakovsky:TPIA:04}, respectively, in the classic Bayesian QCD when $r=1$. The ``$r$-quick convergence'' was originally proposed in \cite{Lai:AnP:76} and has been used in \cite{Tartakovsky:SISP:98, Dragalin:TIT:99} to show the asymptotic optimality of the sequential multi-hypothesis testing. The ``average-$r$-quick convergence'' was originally introduced in \cite{Tartakovsky:TPIA:04} to show asymptotic optimality of the SR procedure in the Bayesian QCD problem.
\end{rmk}

We discuss the conclusion obtained in Theorem \ref{thm:general_opt} in the following. First, we note that if the true post-change parameter $\lambda$ belongs to $\Lambda_{D}$, we have $\min_{\lambda_{i} \in \Lambda_{D}} D(f_{\lambda}, f_{\lambda_{i}}) = 0$; then the upper bound of ADD meets the lower bound, which indicates that $\tau_{R}$ and $\tau_{C}$ are asymptotically optimal for the parameters within $\Lambda_{D}$. As a result, we have the following corollary: 
\begin{cor} 
When $\Lambda$ is a finite set, by choosing $\Lambda_{D} = \Lambda$ and setting $B_{i} = \ldots = B_{I-1} = I(\rho\alpha)^{-1}$, $\tau_{R}$ and $\tau_{C}$ are first order asymptotically optimal as $\alpha \rightarrow 0$. In addition
\begin{eqnarray}
\mathbb{E}_{\pi}^{(\theta_{0}, \lambda_{i})}\left[(\tau_{R}-t)^{+}\right] &\leq& \mathbb{E}_{\pi}^{(\theta_{0}, \lambda_{i})}\left[(\tau_{C}-t)^{+}\right] \no\\
&\leq&\frac{|\log \alpha|}{D(f_{\lambda_{i}}, g_{\theta_{0}})+|\log(1-\rho)|}(1+o(1)) 
\end{eqnarray}
for all $\lambda_{i} \in \Lambda$.
\end{cor}
For the general case $\Lambda \neq \Lambda_{D}$, by Theorem \ref{lem:LB} and Theorem \ref{thm:general_opt}, we have
\begin{eqnarray}
\eta(\tau_{R}; \lambda) \geq \eta(\tau_{C}; \lambda) \geq 1-\frac{\min_{\lambda_{i} \in \Lambda_{D}}D(f_{\lambda}, f_{\lambda_{i}})}{D(f_{\lambda}, g_{\theta_{0}}) + |\log(1-\rho)|} \no
\end{eqnarray}
for any $\lambda \in \Lambda$. Then, it is straightforward to see that $\tau_{R}$ and $\tau_{C}$ are asymptotically $\epsilon-$optimal if we properly select $\Lambda_{D}$ such that
\begin{eqnarray}
\frac{\min_{\lambda_{i} \in \Lambda_{D}}D(f_{\lambda}, f_{\lambda_{i}})}{D(f_{\lambda}, g_{\theta_{0}}) + |\log(1-\rho)|} \leq \epsilon, \quad \forall \lambda \in \Lambda. \label{eq:e_cond}
\end{eqnarray}
In practice, $\Lambda_{D}$ can be designed \emph{offline} by using \eqref{eq:e_cond} for the implementation of the M-SR procedure and the modified M-SR procedure. If $D(f_{\lambda}, f_{\lambda_{i}})$ satisfies certain properties, the design procedure could be further simplified. Let $c:=\min_{\lambda \in \Lambda}D(f_{\lambda}, g_{\theta_{0}}) + |\log(1-\rho)|$ be a constant, it is easy to see that
$$ \min_{\lambda_{i} \in \Lambda_{D}} D(f_{\lambda}, f_{\lambda_{i}}) \leq c \epsilon  $$
is a stronger condition than \eqref{eq:e_cond}. If $D(f_{\lambda}, f_{\lambda_{i}})$ is a Lipschitz continuous function w.r.t. $\lambda_{i}$, then we have
\begin{eqnarray}
D(f_{\lambda}, f_{\lambda_{i}}) = |D(f_{\lambda}, f_{\lambda_{i}}) - D(f_{\lambda}, f_{\lambda})| \leq K|\lambda_{i} - \lambda|, \label{eq:Lip}
\end{eqnarray}
where $K>0$ is a real constant. \eqref{eq:Lip} provides a simple method to design $\Lambda_{D}$: the observer can first divide $\Lambda$ into a series of disjoint sub-intervals whose length is bounded by $c\epsilon / K$, and then constructs $\Lambda_{D}$ by picking one point in each sub-interval. However, we also point out that the cardinality of $\Lambda_{D}$ is in inverse proportion to $\epsilon$, a larger $I$ will lead to more computation and longer detection delay. The asymptotic result in Theorem \ref{thm:general_opt} holds when $\alpha$ is infinitely smaller than $I^{-1}$.

\begin{exmpl}
Assume the pre-change distribution is $\mathcal{N}(0, \sigma^2)$, the post-change distribution is $\mathcal{N}(\lambda, \sigma^{2})$, and $\lambda \in \Lambda = [0.37, 2.63]$. For $\lambda \in \Lambda$ and $\lambda_{i} \in \Lambda_{D}$, it is easy to verify that
\begin{eqnarray}
D(f_{\lambda}, f_{\lambda_{i}}) = \frac{(\lambda-\lambda_{i})^2}{2\sigma^2}.
\end{eqnarray}
When $\sigma^2 = 1$ and $\rho=0.01$, the observer can control the tolerant loss $\epsilon=0.2$ by setting $\Lambda_{D} = [0.5483, 1.4517]$, which is illustrated in Figure \ref{fig:e_optimal}. 
Figure \ref{fig:e_optimal} shows that the performance of M-SR procedure (the black solid curve) lies within between the lower bound of ADD (the blue dot-dash line) and the bound of $\epsilon-$optimality (the red dash line). In addition, the performance of M-SR procedure meets the lower bound at $\lambda = 0.5483$ and $\lambda = 1.4517$.
\begin{figure}[thb]
\centering
\includegraphics[width=0.5 \textwidth]{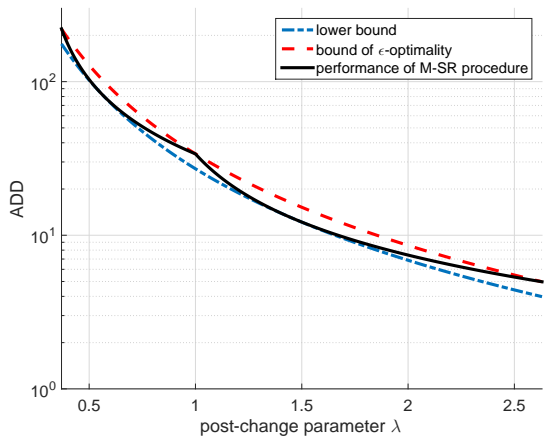}
\caption{$\epsilon-$optimality of M-SR procedure}
\label{fig:e_optimal}
\end{figure}
\end{exmpl}

\section{Extension: Multi-Source Monitoring Problem}\label{sec:ext}
In Section \ref{sec:model} and Section \ref{sec:optimal}, we focus on low complexity algorithms for solving the problem that an abrupt change happens on a single observation sequence. In this section, we extend our study to the multi-source monitoring problem. Particularly, in this section, we assume that there are $L$ mutually independent sources and 
that the observer is capable of monitoring all sources simultaneously; hence the observer obtains a $L-$dimensional random vector at each time slot. We denote the observation sequence as $\{ \mathbf{X}_{1}, \mathbf{X}_{2}, \ldots \}$, in which $\mathbf{X}_{k} = [X_{1, k}, X_{2, k}, \ldots, X_{L, k}]$ is the observation obtained at time slot $k$, and $X_{l, k}$ is the sample from the $l^{th}$ source at time slot $k$. We note that $l \in \{1, \ldots, L\}$ is the source index and $k\in\{1, 2, \ldots\}$ is the time index.

Similar to the model presented in Section \ref{sec:model}, the distribution of the observation sequence changes abruptly at the change-point $t$. The prior distribution of $t$ is the same as \eqref{eq:ext_geometry}. In this section, we further assume that
the distributions of all $L$ sources change simultaneously at $t$. 

Pre-change distributions of the observed sources are perfectly known by the observer but the post-change distributions contain unknown parameters. For the $l^{th}$ source, let $g_{\theta_{l, 0}}(\cdot)$ be the pre-change distribution, in which $\theta_{l, 0}$ is a known parameter. Let $f_{\lambda_{l}}(\cdot)$ be the post-change distribution, which contains an unknown parameter $\lambda_{l}$. Denote $\Lambda_{l}$ as the feasible set of $\lambda_{l}$. Throughout this section, $\Lambda_{l}$ is assumed to be a finite set that can be written as
\begin{eqnarray}
\Lambda_{l} = \{ \lambda_{l, 0}, \lambda_{l, 1}, \ldots, \lambda_{l, I_{l}-1} \}; \no
\end{eqnarray}
hence $|\Lambda_{l}| = I_{l}$. For the scenario that $\Lambda_{l}$ is an interval, the asymptotic optimality result obtained in Theorem \ref{thm:win_tau_c} in the sequel can be easily generalized to the asymptotic $\epsilon-$optimality as we have done in Section \ref{sec:optimal_interval}. Observations from the $l^{th}$ source $\{X_{l, 1}, \ldots\, X_{l, t}, X_{l, t+1}, \ldots\}$ are conditionally i.i.d. (conditioned on the change-point and the post-change parameter).

Therefore, the distribution of $\mathbf{X}_{k}$ contains $L$ parameters. We denote these parameters in a vector form. Specifically, let $\boldsymbol{\theta}_{0} = [\theta_{1,0}, \theta_{2,0}, \ldots, \theta_{L,0}]$ be the parametric vector for pre-change distribution. As our assumption, $\boldsymbol{\theta}_{0}$ is known by observer. Let $\boldsymbol{\lambda}_{u} = [\lambda_{1, i_{1}}, \lambda_{2, i_{2}}, \ldots, \lambda_{L, i_{L}}]$, with $\lambda_{l, i_{l}} \in \Lambda_{l}, l=1, \ldots, L$, be the vector of post-change parameters, which is unknown to the observer. Denote
\begin{eqnarray}
\mathbf{\Lambda} = \Lambda_{1} \times \Lambda_{2} \times \cdots \times \Lambda_{L}
\end{eqnarray}
as the set of all possible post-change parameter vector. Obviously, $|\mathbf{\Lambda}| = \prod_{l=1}^{L}I_{l}$; hence for $\boldsymbol{\lambda}_{u} = [\lambda_{1, i_{1}}, \lambda_{2, i_{2}}, \ldots, \lambda_{L, i_{L}}]$, we have $u \in \{0, 1, \ldots, \prod_{l=1}^{L}I_{l}-1\}$ and $i_{l} \in \{0, \ldots, I_{l}-1\}$. We note that sub-index $u$ and $(i_{1}, \ldots, i_{L})$ have a one-to-one relationship, for example, one may set
$$ u = i_{1} \times I_{2} \times I_{3} \times \cdots \times I_{L} + i_{2} \times I_{3} \times \cdots \times I_{L} + \cdots + i_{L-1}\times I_{L} + i_{L}. $$

Using these notations, the problem we aim to solve can be written as:
\begin{eqnarray}
&&\inf_{\tau \in \mathcal{T}} \mathbb{E}_{\pi}^{(\boldsymbol{\theta}_{0}, \boldsymbol{\lambda})}[(\tau - t)^{+}] \text{ for all } \boldsymbol{\lambda} \in \mathbf{\Lambda}, \no\\
&&\text{subject to } \sup_{\boldsymbol{\lambda} \in \mathbf{\Lambda}} P_{\pi}^{(\boldsymbol{\theta}_{0}, \boldsymbol{\lambda})}(\tau < t) \leq \alpha. \label{eq:P2}
\end{eqnarray}
As a simple extension of the conclusions obtained in Section \ref{sec:optimal}, we know that
\begin{eqnarray}
&& R_{n}^{(u)} := \sum_{k=1}^{n} \prod_{q=k}^{n} \frac{1}{1-\rho} \frac{f_{\boldsymbol{\lambda}_{u}}(X_{1, q}, X_{2, q},\ldots, X_{L, q})}{g_{\boldsymbol{\theta}_{0}}(X_{1, q}, X_{2, q},\ldots, X_{L, q})}, \label{eq:L_Rn}\\
&& \tau_{R} = \inf\left\{n \geq 0: \max_{u \in \{0, 1, \ldots, \prod_{l=1}^{L} I_{l} - 1\}} \log R_{n}^{(u)} \geq \log B \right\}, \label{eq:L_tauR}
\end{eqnarray}
and
\begin{eqnarray}
&& C_{n}^{(u)} := \max_{1\leq k\leq n} \prod_{q=k}^{n} \frac{1}{1-\rho} \frac{f_{\boldsymbol{\lambda}_{u}}(X_{1, q}, X_{2, q},\ldots, X_{L, q})}{g_{\boldsymbol{\theta}_{0}}(X_{1, q}, X_{2, q},\ldots, X_{L, q})}, \label{eq:L_Cn}\\
&& \tau_{C} = \inf\left\{n\geq0: \max_{u \in \{0, 1, \ldots, \prod_{l=1}^{L} I_{l} - 1\}} \log C_{n}^{(u)} \geq \log B\right\} \label{eq:L_tauC}
\end{eqnarray}
are asymptotically optimal if we choose $B=(\rho\alpha)^{-1}\prod_{l=1}^{L}I_{l}$. However, $\tau_{R}$ and $\tau_{C}$ have obvious drawbacks since they need to maintain $\prod_{l=1}^{L}I_{l}$ many detection statistics at each time slot; hence both the computational complexity and the storage complexity increase exponentially w.r.t. $L$; these shortcomings limit the practical implementation of $\tau_{R}$ and $\tau_{C}$ when $L$ is large. In the rest of this section, we propose a \emph{window-based modified M-SR detection procedure}, and show that the proposed algorithm is asymptotically optimal and has low computational complexity. In particular, the proposed algorithm can be written as
\begin{eqnarray}
&& \tilde{C}_{n}^{(u)} := \max_{n-m_{\alpha}\leq k\leq n} \prod_{q=k}^{n} \frac{1}{1-\rho} \frac{f_{\boldsymbol{\lambda}_{u}}(X_{1, q}, X_{2, q},\ldots, X_{L, q})}{g_{\boldsymbol{\theta}_{0}}(X_{1, q}, X_{2, q},\ldots, X_{L, q})}, \label{eq:win_c} \\
&& \tilde{\tau}_{C}^{(u)} = \inf\{n\geq0: \log \tilde{C}_{n}^{(u)} \geq \log B_{u}\}, \label{eq:win_tau_ci} \\
&& \tilde{\tau}_{C} = \min_{u \in \{0, 1, \ldots, \prod_{l=1}^{L} I_{l} - 1\}} \tilde{\tau}_{C}^{(u)}. \label{eq:win_tau_c}
\end{eqnarray}
We note that the proposed procedure is the same as the modified M-SR procedure except that $\tilde{C}_{n}^{(u)}$ takes maximum over the latest $m_{\alpha}$ likelihood ratios instead of all likelihood ratios. The asymptotic optimality of $\tilde{\tau}_{C}$ is presented in the following theorem:
\begin{thm} \label{thm:win_tau_c}
By choosing $m_{\alpha}$ such that
$$\liminf \frac{m_{\alpha}}{|\log \alpha|} > \max_{\boldsymbol{\lambda}_{u} \in \mathbf{\Lambda}} \frac{1}{d(\boldsymbol{\theta}_{0}, \boldsymbol{\lambda}_{u})} \text{  but  } \log m_{\alpha} = o(|\log \alpha|),$$
and setting threshold $B = (\rho\alpha)^{-1}\prod_{l=1}^{L}I_{l}$, then $\tilde{\tau}_{C}$ satisfies the false alarm constraint. In addition, as $\alpha \rightarrow 0$ ,
\begin{eqnarray}
\mathbb{E}_{\pi}^{(\boldsymbol{\theta}_{0}, \boldsymbol{\lambda}_{u})}[(\tilde{\tau}_{C} - t)^{+}] \leq \frac{|\log B|}{d(\boldsymbol{\theta}_{0}, \boldsymbol{\lambda}_{u})}(1+o(1)), \text{ for all } \boldsymbol{\lambda}_{u} \in \mathbf{\Lambda} \label{eq:L_delay}
\end{eqnarray}
in which
\begin{eqnarray}
d(\boldsymbol{\theta}_{0}, \boldsymbol{\lambda}_{u}) = D(f_{\boldsymbol{\lambda}_{u}}, g_{\boldsymbol{\theta}_{0}}) + |\log(1-\rho)| = \sum_{l=1}^{L} D(f_{\lambda_{l, i_{l}}}, g_{\theta_{l, 0}}) + |\log(1-\rho)|.
\end{eqnarray}
\end{thm}
\begin{proof}
Please see Section \ref{sec:proof_IV}.
\end{proof}
\begin{rmk}
Similar to the result presented in Theorem \ref{lem:LB}, the lower bound of ADD when the observer monitors $L$ sources simultaneously is given as
\begin{eqnarray}
\inf_{\tau \in \mathcal{T}} \mathbb{E}_{\pi}^{(\boldsymbol{\theta}_{0}, \boldsymbol{\lambda}_{u})}[(\tau - t)^{+}] \geq \frac{|\log B|}{d(\boldsymbol{\theta}_{0}, \boldsymbol{\lambda}_{u})}(1+o(1)); \label{eq:multi_LB}
\end{eqnarray}
In \eqref{eq:L_delay}, we note $|\log B| = |\log \alpha| + |\log \rho| + \sum_{l=1}^{L}\log I_{l}$; hence when $\rho$, $L$ and $I_{l}$'s are fixed, as $\alpha \rightarrow 0$, \eqref{eq:L_delay} indicates that $\tilde{\tau}_{C}$ is first-order asymptotically optimal.
\end{rmk}


Concise expressions \eqref{eq:win_c} to \eqref{eq:win_tau_c} bring mathematical convenience in showing the asymptotic performance of $\tilde{\tau}_{C}$. In the following, we rewrite the proposed algorithm into another equivalent form for the implementation purpose. We note that \eqref{eq:win_tau_c} can be equivalently written as
\begin{eqnarray}
\tilde{\tau}_{C} = \inf\left\{ n\geq1: \max_{\boldsymbol{\lambda}_{u} \in \mathbf{\Lambda}}(\tilde{C}_{n}^{(u)} - B_{u}) \geq 0 \right\}.
\end{eqnarray}
As indicated in Theorem \ref{thm:win_tau_c}, $B_{u}$ can be chosen as a constant over $u$ to guarantee the asymptotic optimality. In this case, we have
\begin{eqnarray}
\max_{\boldsymbol{\lambda}_{u} \in \mathbf{\Lambda}} \tilde{C}_{n}^{(u)} &=& \max_{\boldsymbol{\lambda}_{u} \in \mathbf{\Lambda}} \max_{n-m_{\alpha}\leq k\leq n} \prod_{q=k}^{n} \frac{1}{1-\rho} \frac{f_{\boldsymbol{\lambda}_{u}}(X_{1, q}, X_{2, q},\ldots, X_{L, q})}{g_{\boldsymbol{\theta}_{0}}(X_{1, q}, X_{2, q},\ldots, X_{L, q})} \no\\
&=& \max_{n-m_{\alpha}\leq k\leq n} \left(\frac{1}{1-\rho}\right)^{n-k+1} \max_{\boldsymbol{\lambda}_{u} \in \mathbf{\Lambda}} \prod_{q=k}^{n} \frac{f_{\boldsymbol{\lambda}_{u}}(X_{1, q}, X_{2, q},\ldots, X_{L, q})}{g_{\boldsymbol{\theta}_{0}}(X_{1, q}, X_{2, q},\ldots, X_{L, q})} \no\\
&=& \max_{n-m_{\alpha}\leq k\leq n} \left(\frac{1}{1-\rho}\right)^{n-k+1} \max_{(\lambda_{1, i_{1}}, \lambda_{2, i_{2}}, \ldots, \lambda_{L, i_{L}}) \in \mathbf{\Lambda}} \prod_{q=k}^{n} \prod_{l=1}^{L} \frac{f_{\lambda_{l, i_{l}}}(X_{l, q})}{g_{\theta_{l, 0}}(X_{l, q})} \no\\
&=& \max_{n-m_{\alpha}\leq k\leq n} \left(\frac{1}{1-\rho}\right)^{n-k+1} \prod_{l=1}^{L} \left[ \max_{\lambda_{l, i_{l}} \in \Lambda_{l}} \prod_{q=k}^{n}  \frac{f_{\lambda_{l, i_{l}}}(X_{l, q})}{g_{\theta_{l, 0}}(X_{l, q})} \right]. \label{eq:win_C_imp}
\end{eqnarray}
By taking logarithm on both sides, we have
\begin{eqnarray}
\max_{\boldsymbol{\lambda}_{u} \in \mathbf{\Lambda}} \log \tilde{C}_{n}^{(u)} = \max_{n-m_{\alpha}\leq k\leq n} \left\{ (n-k+1)|\log (1-\rho)| + \sum_{l=1}^{L} \left[ \max_{\lambda_{l, i_{l}} \in \Lambda_{l}} \sum_{q=k}^{n} \log \frac{f_{\lambda_{l, i_{l}}}(X_{l, q})}{g_{\theta_{l, 0}}(X_{l, q})} \right] \right\}. \label{eq:win_log_C_imp}
\end{eqnarray}
\eqref{eq:win_C_imp} or \eqref{eq:win_log_C_imp} indicates that the problem of searching best $\boldsymbol{\lambda}_{u} \in \mathbf{\Lambda}$ from $\prod_{l=1}^{L} I_{l}$ elements can be decomposed into $L$ sub-problems, each of which only needs to find the best $\lambda_{l, i_{l}} \in \Lambda_{l}$ from $I_{l}$ elements. Hence the computational complexity substantially decrease from $\prod_{l=1}^{L} I_{l}$ to $\sum_{l=1}^{L} I_{l}$. In the following, we discuss in detail about the implementation and the complexity of the above window-based modified M-SR algorithm.

For the proposed algorithm, the main computation lies on the term
\begin{eqnarray}
\max_{\lambda_{l, i_{l}} \in \Lambda_{l}} \sum_{q=k}^{n} \log \frac{f_{\lambda_{l, i_{l}}}(X_{l, q})}{g_{\theta_{l, 0}}(X_{l, q})}. \label{eq:computation_unit}
\end{eqnarray}
Note that we need to solve \eqref{eq:computation_unit} for $k = n-m_{\alpha}, \ldots, n$. These problems could be difficult to solve in general since $\Lambda_{l}$ is a discrete set and $f_{\lambda_{l, i_{l}}}$, $g_{\theta_{l, 0}}$ could lead to a rather complex log-likelihood formula. To solve \eqref{eq:computation_unit}, the observer can maintain a structure, which is termed as a \emph{unit} in this section, as illustrated in Figure \ref{fig:computation_unit}. In particular, the $l^{th}$ unit, which corresponds to the $l^{th}$ source, requires two storage spaces: the first space is named as the LLR table with $I_{l}$ rows and $m_{\alpha}+1$ columns. Specifically, the $i^{th}$ row, $i\in\{0, 1, \ldots, I_{I}-1\}$, and the $k^{th}$ column, $k\in\{n-m_{\alpha}, \ldots, n\}$, stores the value of $\sum_{q=k}^{n} \log \frac{f_{\lambda_{l, i_{l}}}(X_{l, q})}{g_{\theta_{l, 0}}(X_{l, q})}$. The second space with size $1 \times (m_{\alpha}+1)$ is named as maximum value container, the $k^{th}$ cell of which stores the maximum value the $k^{th}$ column of the LLR table; hence the maximum value container stores the optimal value of \eqref{eq:computation_unit} for $k = n-m_{\alpha}, \ldots, n$.

\begin{figure}[thb]
\centering
\includegraphics[width=0.55 \textwidth]{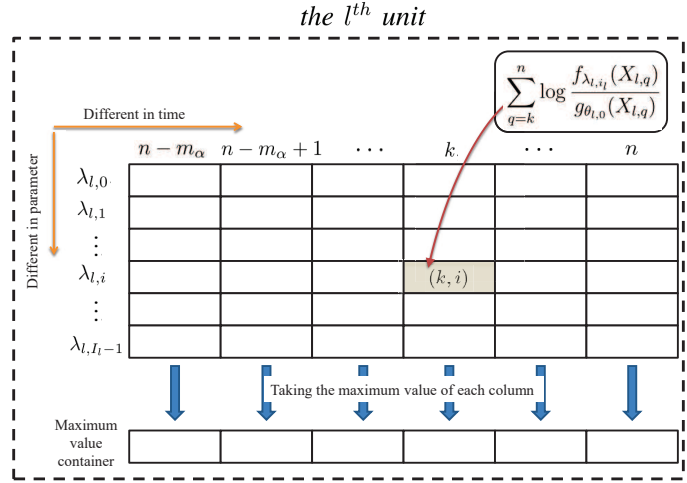}
\caption{The $l^{th}$ unit consists of two parts: 1) a table to store log likelihood ratios with size $I_{l}\times (m_{\alpha}+1)$, and 2) a container to store maximum values with size $1\times (m_{\alpha}+1)$. }
\label{fig:computation_unit}
\end{figure}

The LLR table is updated at each time slot whenever there comes a new observation. Assume the observer newly receives $\mathbf{X}_{n+1}$, then he updates the $l^{th}$ LLR table by following steps: 1) selecting component $X_{l, n+1}$ in $\mathbf{X}_{n+1}$ and computing the value of $\log \frac{f_{\lambda_{l, i}}(X_{l, n+1})}{g_{\theta_{l, 0}}(X_{l, n+1})}$ for all $i \in \{ 0, \ldots, I_{l}-1\}$; 2) erasing the values in the left-most column in the $l^{th}$ LLR table, shifting the rest of values one column to the left and then writing zeros in the right-most column; 3) adding $\log \frac{f_{\lambda_{l, i}}(X_{l, n+1})}{g_{\theta_{l, 0}}(X_{l, n+1})}$ to the $(i+1)^{th}$ row of the LLR table. This updating procedure simply uses the recursive relationship
\begin{eqnarray}
\sum_{q=k}^{n+1} \log \frac{f_{\lambda_{l, i_{l}}}(X_{l, q})}{g_{\theta_{l, 0}}(X_{l, q})} = \sum_{q=k}^{n} \log \frac{f_{\lambda_{l, i_{l}}}(X_{l, q})}{g_{\theta_{l, 0}}(X_{l, q})} + \log \frac{f_{\lambda_{l, i_{l}}}(X_{l, n+1})}{g_{\theta_{l, 0}}(X_{l, n+1})} \no
\end{eqnarray}
of the statistic in \eqref{eq:computation_unit}. Once the LLR table has been updated, the maximum value container needs to be updated correspondingly. Based on above discussions, it is easy to see that the computational complexity of updating the $l^{th}$ unit is on the order of $O(m_{\alpha}I_{l})$ since both updating the LLR table and updating the maximum value container require $O(m_{\alpha}I_{l})$ computational amount.

\begin{figure}[thb]
\centering
\includegraphics[width=0.75 \textwidth]{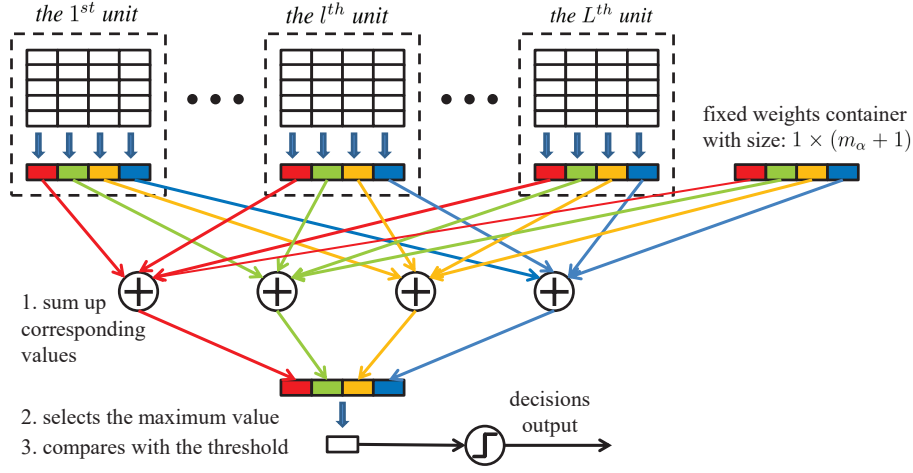}
\caption{The implementation of the window-based modified M-SR procedure}
\label{fig:MSR_structure}
\end{figure}

According to \eqref{eq:win_log_C_imp}, Figure \ref{fig:MSR_structure} illustrates a structure for the implementation of the proposed algorithm. Besides $L$ units for $L$ monitoring sources described above, we require another storage space with size $1 \times (m_{\alpha}+1)$ to store the value of $(n-k+1)|\log(1-\rho)|$ for $k\in\{n-m_{\alpha}, \ldots, n\}$. We term this space as the fixed weights container, which is illustrated in the upper-right corner in Figure \ref{fig:MSR_structure}. Noting that both the maximum value containers within the units and the fixed weights container have the same size, the detection statistic $\tilde{C}_{n}$ can be calculated by summing up the values in the cells with the same index $k$ and then selects the maximum of them. Finally, the decision rule can be made by comparing the value of $\tilde{C}_{n}$ with our pre-designed threshold.

Given the values in the maximum value containers within units, it is easy to see that the computational amount of calculating $\tilde{C}_{n}$ is on the order of $O(m_{\alpha}L)$. Hence, the main computational burden of implementing the whole algorithm is located in the updates of $L$ units, and the total computational amount is then on the order of $O(m_{\alpha} \sum_{l=1}^{L}I_{l})$. However, Figure \ref{fig:MSR_structure} indicates that the calculation procedure can be speed up by updating $L$ units in a parallel manner. This is because the update of the $l^{th}$ unit only requires $X_{l,n}$ in $\mathbf{X}_{n}$, which is a consequence of the mutual independence between each two sources. The storage requirement, as the structure illustrated in Figure \ref{fig:MSR_structure}, can be easily found to be $O(m_{\alpha}\sum_{l=1}^{L}I_{l})$.

\section{Proof} \label{sec:proof}
\subsection{Proofs in Section \ref{sec:optimal}} \label{subsec:optimal}
In this subsection, we prove the main conclusion, Theorem \ref{thm:general_opt}, in Section \ref{sec:optimal}. To proceed, we first present some supporting lemmas.

\begin{lem} \label{lem:Pfa}
By setting $B_{0} = \ldots = B_{I-1} = I(\rho\alpha)^{-1}$, the M-SR procedure $\tau_{R}$ defined in \eqref{eq:tau_r} satisfies the false alarm probability constraint, i.e.,
\begin{eqnarray}
\sup_{\lambda\in\Lambda}\mathrm{PFA}(\tau_{R}; \theta_{0}, \lambda) \leq \alpha.
\end{eqnarray}
\end{lem}
\begin{proof}
We note that $\Lambda$ is the compact set of all possible post-change parameter and $\Lambda_{D} \subseteq \Lambda$ is the finite set to implement $\tau_{R}$. 
From \eqref{eq:Lambda1}, we have
\begin{eqnarray}
\log R_{n}^{(\theta_{0}, \lambda_{i})} = \log \frac{\pi_{n}^{(\theta_{0}, \lambda_{i})}}{1-\pi_{n}^{(\theta_{0}, \lambda_{i})}} - \log \rho.
\end{eqnarray}
As a result setting $B_{i} = I(\rho\alpha)^{-1}$, $\tau_{R}^{(i)}$ defined in \eqref{eq:tau_ri} can be written as
\begin{eqnarray}
\tau_{R}^{(i)} &=& \inf\left\{ n \geq 1 | \log R_{n}^{(\theta_{0}, \lambda_{i})} > \log B_{i} \right\} \no\\
&=& \inf\left\{ n \geq 1 \Big| \pi_{n}^{(\theta_{0}, \lambda_{i})} > 1-\left(1+\frac{\alpha}{I}\right)^{-1} \frac{\alpha}{I} \right\}.
\end{eqnarray}
Hence
\begin{eqnarray}
P_{\pi}^{(\theta_{0}, \lambda_{i})}\left(\tau_{R} < t; \tau_{R} = \tau_{R}^{(i)}\right) \leq P_{\pi}^{(\theta_{0}, \lambda_{i})}\left( \tau_{R}^{(i)} < t \right) = \mathbb{E}_{\pi}^{(\theta_{0}, \lambda_{i})}\left[1-\pi_{\tau_{R}^{(i)}}^{(\theta_{0}, \lambda_{i})}\right] \leq \alpha/I. \no
\end{eqnarray}
We note that
\begin{eqnarray}
P_{\pi}^{(\theta_{0}, \lambda_{j})}\left(\tau_{R} < t; \tau_{R} = \tau_{R}^{(i)}\right) = P_{\pi}^{(\theta_{0}, \lambda_{i})}\left(\tau_{R} < t; \tau_{R} = \tau_{R}^{(i)}\right). \no
\end{eqnarray}
This is because all observations on the event $\{ \tau_{R} < t \}$ are generated from $g_{\theta_0}(x)$, which is not related to the post-change parameter. As a result, we have
\begin{eqnarray}
P_{\pi}^{(\theta_{0}, \lambda_{j})}(\tau_{R} < t) = \sum_{i=0}^{I-1} P_{\pi}^{(\theta_{0}, \lambda_{j})}\left(\tau_{R} < t; \tau_{R} = \tau_{R}^{(i)}\right) \leq \alpha. \label{eq:fa1}
\end{eqnarray}
Using above argument again, we conclude
\begin{eqnarray}
P_{\pi}^{(\theta_{0}, \lambda)}(\tau_{R} < t) = P_{\pi}^{(\theta_{0}, \lambda_{j})}(\tau_{R} < t) \leq \alpha
\end{eqnarray}
for all $\lambda \in \Lambda$. As a result, we have $\sup_{\lambda \in \Lambda}P_{\pi}^{(\theta_{0}, \lambda)}(\tau_{R} < t) \leq \alpha $.
\end{proof}

In the sequel, Lemma \ref{lem:Pfa} will be used in proving the false alarm probability of $\tau_{R}$ and $\tau_{C}$, and the following Lemma will be used in bounding the detection delay. To proceed, we define some notations. Let
\begin{eqnarray}
l(x_{q}; \theta_{0}, \lambda_{i}) := \log \frac{f_{\lambda_{i}}(x_{q})}{g_{\theta_{0}}(x_{q})}
\end{eqnarray}
be the log likelihood ratio, and define a new statistic as
\begin{eqnarray}
S_{m:n}^{(\theta_{0}, \lambda_{i})} := (n-m+1)|\log(1-\rho)| + \sum_{q=m}^{n} l(x_{q}; \theta_{0}, \lambda_{i}). \label{eq:stat_S}
\end{eqnarray}
Furthermore, let
\begin{eqnarray}
&&\tau_{S}^{(i)} := \inf\left\{ n \geq 1: S_{1:n}^{(\theta_{0}, \lambda_{i})} \geq \log B_{i} \right\}, \label{eq:tau_si} \\
&&\tau_{S} = \min_{ i \in \{0, \ldots, I-1\}} \tau_{S}^{(i)}. \label{eq:tau_s}
\end{eqnarray}
Then, we have the following lemma.
\begin{lem} \label{lem:Add}
With Assumptions \eqref{eq:quickly_converge} and \eqref{eq:averege_quickly}, as $B_{i} \rightarrow \infty$, we have
\begin{eqnarray}
&&\hspace{-20mm}\mathbb{E}_{\pi}^{(\theta_{0}, \lambda)}\left[\tau_{S}^{(i)} - t |\tau_{S}^{(i)}  \geq t \right] \no\\
&\leq& \frac{\log B_{i}}{D(f_{\lambda}, g_{\theta_0}) - D(f_{\lambda}, f_{\lambda_i}) + |\log(1-\rho)|}(1+o(1)).
\end{eqnarray}
\end{lem}
\begin{proof}
On the event $\{ t = k\}$, we can decompose $S_{1:n}^{(\theta_{0}, \lambda_{i})}$ into two parts if $n \geq k$:
\begin{eqnarray}
S_{1:n}^{(\theta_{0}, \lambda_{i})} = S_{1:k-1}^{(\theta_{0}, \lambda_{i})} + S_{k:n}^{(\theta_{0}, \lambda_{i})}. \label{eq:twoS}
\end{eqnarray}
When $\lambda$ is the true post-change parameter, by the strong law of large numbers, we have
\begin{eqnarray}
\frac{1}{r}S_{k:k+r-1}^{(\theta_{0}, \lambda_{i})} \overset{a.s.}{\rightarrow} \mathbb{E}_{\lambda}\left[ \log \frac{f_{\lambda_{i}}(X)}{g_{\theta_{0}}(X)} + |\log(1-\rho)| \right] =: d(\theta_{0}, \lambda_{i}; \lambda), \label{eq:as_converge}
\end{eqnarray}
in which
$$d(\theta_{0}, \lambda_{i}; \lambda) = D(f_{\lambda}, g_{\theta_0}) - D(f_{\lambda}, f_{\lambda_{i}}) + |\log(1-\rho)|.$$
By \eqref{eq:twoS}, $\tau_{S}^{(i)} $ can be written equivalently as
\begin{eqnarray}
\tau_{S}^{(i)}  = \inf\left\{j>0: S_{k:j}^{(\theta_{0}, \lambda_{i})} \geq \log B_{i} - S_{1:k-1}^{(\theta_{0}, \lambda_{i})} \right\}. \no
\end{eqnarray}
Hence,
\begin{eqnarray}
S_{k:\tau_{S}^{(i)}-1}^{(\theta_{0}, \lambda_{i})} < \log B_{i} - S_{1:k-1}^{(\theta_{0}, \lambda_{i})}.
\end{eqnarray}
Define the random variable
\begin{eqnarray}
T_{\lambda_{i}, \varepsilon}^{(k, \lambda)} := \sup\left\{ n \geq 1: \big|n^{-1}S_{k:k+n-1}^{(\theta_{0}, \lambda_{i})} - d(\theta_{0}, \lambda_{i}; \lambda)\big| > \varepsilon \right\}.\no
\end{eqnarray}
By \eqref{eq:as_converge}, $T_{\lambda_{i}, \varepsilon}^{(k, \lambda)} < \infty$ almost surely under probability measure $P_{k}^{(\theta_{0}, \lambda)}$. In addition, $\mathbb{E}_{k}^{(\theta_{0}, \lambda)}\left[T_{\lambda_{i}, \varepsilon}^{(k, \lambda)}\right] < \infty$ and $\mathbb{E}_{\pi}^{(\theta_{0}, \lambda)}\left[T_{\lambda_{i}, \varepsilon}^{(\lambda)}\right] < \infty$ by Assumption \eqref{eq:quickly_converge} and \eqref{eq:averege_quickly}.

On the event $\left\{ \tau_{S}^{(i)}  > T_{\lambda_{i}, \varepsilon}^{(k, \lambda)} + (k-1) \right\}$, we have
\begin{eqnarray}
S_{k:\tau_{S}^{(i)}-1}^{(\theta_{0}, \lambda_{i})} > (\tau_{S}^{(i)}-k+1)(d(\theta_{0}, \lambda_{i}; \lambda)-\varepsilon), \no
\end{eqnarray}
hence
\begin{eqnarray}
\tau_{S}^{(i)} -k+1 < \frac{S_{k:\tau_{S}^{(i)}-1}^{(\theta_{0}, \lambda_{i})}}{d(\theta_{0}, \lambda_{i}; \lambda)-\varepsilon} < \frac{\log B_{i}-S_{1:k-1}^{(\theta_{0}, \lambda_{i})}}{d(\theta_{0}, \lambda_{i}; \lambda)-\varepsilon}.
\end{eqnarray}
Then we have
\begin{eqnarray}
\tau_{S}^{(i)}-k+1 &<& \frac{\log B_{i}-S_{1:k-1}^{(\theta_{0}, \lambda_{i})}}{d(\theta_{0}, \lambda_{i}; \lambda)-\varepsilon} \mathbf{1}_{\left\{ \tau_{S}^{(i)} > T_{\lambda_{i}, \varepsilon}^{(k, \lambda)} + (k-1) \right\}} + T_{\lambda_{i}, \varepsilon}^{(k, \lambda)} \mathbf{1}_{ \left\{\tau_{S}^{(i)} \leq T_{\lambda_{i}, \varepsilon}^{(k, \lambda)} + (k-1) \right\}} \no\\
&<& \frac{\log B_{i}-S_{1:k-1}^{(\theta_{0}, \lambda_{i})}}{d(\theta_{0}, \lambda_{i}; \lambda)-\varepsilon}+ T_{\lambda_{i}, \varepsilon}^{(k, \lambda)}. \no
\end{eqnarray}
Taking the conditional expectation on both sides, since $\mathbb{E}_{k}^{(\theta_{0}, \lambda)}\left[T_{\lambda_{i}, \varepsilon}^{(k, \lambda)}\right] < \infty$ and $P_{k}^{(\theta_{0}, \lambda)}\left(\tau_{S}^{(i)} \geq k \right) \rightarrow 1$ as $B_{i} \rightarrow \infty$, then we have
\begin{eqnarray}
&&\mathbb{E}_{k}^{(\theta_{0}, \lambda)}[\tau_{S}^{(i)}-k|\tau_{S}^{(i)}\geq k] \no\\
&&\leq \frac{\log B_{i}}{d(\theta_{0}, \lambda_{i}; \lambda)-\varepsilon} - \frac{\mathbb{E}_{k}^{(\theta_{0}, \lambda)}[S_{1:k-1}^{(\theta_{0}, \lambda_{i})}|\tau_{S}^{(i)} \geq k]}{d(\theta_{0}, \lambda_{i}; \lambda)-\varepsilon} + \mathbb{E}_{k}^{(\theta_{0}, \lambda)}[T_{\lambda_{i}, \varepsilon}^{(k, \lambda)}|\tau_{S}^{(i)} \geq k] \no \\
&&= \frac{\log B_{i}}{d(\theta_{0}, \lambda_{i}; \lambda)-\varepsilon}(1+o(1)) - \frac{\mathbb{E}_{k}^{(\theta_{0}, \lambda)}[S_{1:k-1}^{(\theta_{0}, \lambda_{i})}|\tau_{S}^{(i)}\geq k]}{d(\theta_{0}, \lambda_{i}; \lambda)-\varepsilon}. \no
\end{eqnarray}
Therefore,
\begin{eqnarray}
&& \mathbb{E}_{\pi}^{(\theta_{0}, \lambda)}[\tau_{S}^{(i)}-t|\tau_{S}^{(i)}\geq t] \no\\
&&= \frac{1}{P_{\pi}^{(\theta_{0}, \lambda)}(\tau_{S}^{(i)}\geq t)}\mathbb{E}_{\pi}^{(\theta_{0}, \lambda)}[\tau_{S}^{(i)}-t; \tau_{S}^{(i)}\geq t] \no\\
&&=\frac{1}{P_{\pi}^{(\theta_{0}, \lambda)}(\tau_{S}^{(i)}\geq t)} \sum_{k=1}^{\infty}P(t=k)\mathbb{E}_{k}^{(\theta_{0}, \lambda)}[\tau_{S}^{(i)} -k|\tau_{S}^{(i)} \geq k]P_{k}^{(\theta_{0}, \lambda)}(\tau_{S}^{(i)}\geq k) \no\\
&&\leq \frac{\log B_{i}}{d(\theta_{0}, \lambda_{i}; \lambda)-\varepsilon}(1+o(1)) - \frac{\mathbb{E}_{\pi}^{(\theta_{0}, \lambda)}\left[S_{1:t-1}^{(\theta_{0}, \lambda_{i})}|\tau_{S}^{(i)} \geq t\right]}{d(\theta_{0}, \lambda_{i}; \lambda)-\varepsilon}.  \label{eq:result}
\end{eqnarray}
We note that $\mathbb{E}_{k}^{(\theta_{0}, \lambda)}\left[S_{1:k-1}^{(\theta_{0}, \lambda_{i})}|\tau_{S}^{(i)} \geq k \right]$ and $\mathbb{E}_{\pi}^{(\theta_{0}, \lambda)}\left[S_{1:t-1}^{(\theta_{0}, \lambda_{i})}|\tau_{S}^{(i)}\geq t\right]$ are finite. To see this,
\begin{eqnarray}
\mathbb{E}_{k}^{(\theta_{0}, \lambda)}\left[S_{1:k-1}^{(\theta_{0}, \lambda_{i})}\right] &\overset{(a)}=& \mathbb{E}_{\infty}\left[S_{1:k-1}^{(\theta_{0}, \lambda_{i})}\right] \no\\
&=& \mathbb{E}_{\infty}\left[\sum_{q=1}^{k-1}l(X_{q}; \theta_{0}, \lambda_{i}) \right] + (k-1)|\log(1-\rho)| \no\\
&=& -(k-1)D(g_{\theta_0}, f_{\lambda}) + (k-1)|\log(1-\rho)|, \no
\end{eqnarray}
where (a) is true because $\{ X_{1}, \ldots, X_{k-1} \}$ are generated by $g_{\theta_{0}}$. Hence, we have
\begin{eqnarray}
-k D(g_{\theta_{0}}, f_{\lambda}) < \mathbb{E}_{k}^{(\theta_{0}, \lambda)}\left[S_{1:k-1}^{(\theta_{0}, \lambda_{i})}\right] < k|\log(1-\rho)|. \no
\end{eqnarray}
Since
\begin{eqnarray}
\mathbb{E}_{\pi}^{(\theta_{0}, \lambda)}[S_{1:t-1}^{(\theta_{0}, \lambda_{i})}] = \sum_{k=1}^{\infty} \mathbb{E}_{k}^{(\theta_{0}, \lambda)}\left[S_{1:k-1}^{(\theta_{0}, \lambda_{i})}\right] P(t=k), \no
\end{eqnarray}
we have
\begin{eqnarray}
-\frac{D(g_{\theta_{0}}, f_{\lambda})}{1-\rho} < \mathbb{E}_{\pi}^{(\theta_{0}, \lambda)}\left[S_{1:t-1}^{(\theta_{0}, \lambda_{i})}\right] < \frac{|\log(1-\rho)|}{1-\rho}. \no
\end{eqnarray}
Therefore, $\mathbb{E}_{\pi}^{(\theta_{0}, \lambda)}\left[S_{1:k-1}^{(\theta_{0}, \lambda_{i})}\right]$ is bounded. 
Since $ P_{\pi}^{(\theta_{0}, \lambda)}\left( \tau_{S}^{(i)} \geq t \right) \rightarrow 1$ as $B_{i} \rightarrow \infty$, then
$$\mathbb{E}_{\pi}^{(\theta_{0}, \lambda)}\left[S_{1:t-1}^{(\theta_{0}, \lambda_{i})}|\tau_{S}^{(i)}\geq t\right] \rightarrow \mathbb{E}_{\pi}^{(\theta_{0}, \lambda)}\left[S_{1:t-1}^{(\theta_{0}, \lambda_{i})}\right] \text{ as } B_{i} \rightarrow \infty.$$
By \eqref{eq:result} we obtain
\begin{eqnarray}
\mathbb{E}_{\pi}^{(\theta_{0}, \lambda)}\left[\tau_{S}^{(i)} -t \big|\tau_{S}^{(i)} \geq t \right] \leq \frac{\log B_{i}}{d(\theta_{0}, \lambda_{i}; \lambda)-\varepsilon}(1+o(1)).
\end{eqnarray}
Since the above equation holds for any $\varepsilon > 0$, then
\begin{eqnarray}
\mathbb{E}_{\pi}^{(\theta_{0}, \lambda)}\left[\tau_{S}^{(i)} -t \big|\tau_{S}^{(i)} \geq t \right] \leq \frac{\log B_{i}}{d(\theta_{0}, \lambda_{i}; \lambda)}(1+o(1)). \no
\end{eqnarray}
\end{proof}



\emph{Proof of Theorem \ref{thm:general_opt}}

Recall the definitions of $R_{n}^{(\theta_{0}, \lambda_{i})}$ in \eqref{eq:stat_R}, $C_{n}^{(\theta_{0}, \lambda_{i})}$ in \eqref{eq:stat_C} and $S_{1:n}^{(\theta_{0}, \lambda_{i})}$ in \eqref{eq:stat_S},
we have $\log R_{n}^{(\theta_{0}, \lambda_{i})} \geq \log C_{n}^{(\theta_{0}, \lambda_{i})} \geq S_{1:n}^{(\theta_{0}, \lambda_{i})}$. Hence for the same threshold $B_{i}$, we have $\tau_{R}^{(i)} \leq \tau_{C}^{(i)} \leq \tau_{S}^{(i)}$, which further indicates $\tau_{R} \leq \tau_{C} \leq \tau_{S}$. Therefore, it is sufficient for us to bound the false alarm probability of $\tau_{R}$ and to bound the average detection delay of $\tau_{S}$.

Specifically, we set threshold $B_{0} = \ldots = B_{I-1}=I(\rho\alpha)^{-1}$. Then, for the false alarm probability, we have
\begin{eqnarray}
\sup_{\lambda \in \Lambda}P_{\pi}^{(\theta_{0}, \lambda)}(\tau_{C} < t) \leq \sup_{\lambda \in \Lambda}P_{\pi}^{(\theta_{0}, \lambda)}(\tau_{R} < t) \leq \alpha,
\end{eqnarray}
in which the last inequality is because of Lemma \ref{lem:Pfa}. Therefore, both $\tau_{R}$ and $\tau_{C}$ satisfies the false alarm constraint. Furthermore, for ADD
\begin{eqnarray}
\mathbb{E}_{\pi}^{(\theta_{0}, \lambda)}[(\tau_{R} - t)^{+}] &\leq& \mathbb{E}_{\pi}^{(\theta_{0}, \lambda)}[(\tau_{C} - t)^{+}] \leq \mathbb{E}_{\pi}^{(\theta_{0}, \lambda)}[(\tau_{S} - t)^{+}] \leq \mathbb{E}_{\pi}^{(\theta_{0}, \lambda)}\left[(\tau_{S}^{(i)} - t)^{+}\right] \no\\
&=& P\left(\tau_{S}^{(i)} \geq t \right)\mathbb{E}_{\pi}^{(\theta_{0}, \lambda)}\left[\tau_{S}^{(i)} - t \big| \tau_{S}^{(i)} \geq t \right] \no\\
&\leq& \frac{|\log \alpha|}{D(f_{\lambda}, g_{\theta_0}) - D(f_{\lambda}, f_{\lambda_i}) + |\log(1-\rho)|}(1+o(1)), \label{eq:local1}
\end{eqnarray}
in which the last inequality is because $P\left(\tau_{S}^{(i)} \geq t \right) \rightarrow 1$ as $\alpha \rightarrow 0$ and the conclusion obtained in Lemma \ref{lem:Add}. Note that \eqref{eq:local1} holds for all $\lambda_{i} \in \Lambda_{D}$, then we can choose a smallest upper bound as
\begin{eqnarray}
\mathbb{E}_{\pi}^{(\theta_{0}, \lambda)}[(\tau_{R} - t)^{+}] &\leq& \mathbb{E}_{\pi}^{(\theta_{0}, \lambda)}[(\tau_{C} - t)^{+}] \no\\
&\leq& \frac{|\log \alpha|}{D(f_{\lambda}, g_{\theta_0}) - \max_{\lambda_{i}\in\Lambda_{D}}D(f_{\lambda}, f_{\lambda_i}) + |\log(1-\rho)|}(1+o(1)). \no
\end{eqnarray}
That ends the proof.

\subsection{Proofs in Section \ref{sec:ext}} \label{sec:proof_IV}
\emph{Proof of Theorem \ref{thm:win_tau_c}}

We first show that $\tilde{\tau}_{C}$ satisfies the false alarm constraint. Recall \eqref{eq:L_Rn}, \eqref{eq:L_Cn} and \eqref{eq:win_c}, we have $\tilde{C}^{(u)}_{n} \leq C^{(u)}_{n} \leq R^{(u)}_{n}$ by definition, therefore $\tilde{\tau}_{C} \geq \tau_{R}$ and $P_{\pi}^{(\boldsymbol{\theta}_{0}, \boldsymbol{\lambda})}(\tilde{\tau}_{C} < t) \leq P_{\pi}^{(\boldsymbol{\theta}_{0}, \boldsymbol{\lambda})}(\tau_{R} < t)$. Following the similar argument presented in Lemma \ref{lem:Pfa}, we can show that $\sup_{\boldsymbol{\lambda} \in \mathbf{\Lambda}} P_{\pi}^{(\boldsymbol{\theta}_{0}, \boldsymbol{\lambda})}(\tau_{R} < t) \leq \alpha$ by choosing $B_{u}=(\alpha\rho)^{-1}\prod_{l=1}^{L}I_{l}$. Hence the false alarm constraint is satisfied.

We then analyze the detection delay. Recall that
\begin{eqnarray}
d(\boldsymbol{\theta}_{0}, \boldsymbol{\lambda}_{u}) = D(f_{\boldsymbol{\lambda}_{u}}, g_{\boldsymbol{\theta}_{0}}) + |\log(1-\rho)| = \sum_{l=1}^{L} D(f_{\lambda_{l, i_{l}}}, g_{\theta_{l, 0}}) + |\log(1-\rho)|. \no
\end{eqnarray}
By the strong law of large number, when $\boldsymbol{\lambda}_{u}$ is the true post-change parameter, on the event $\{t=k\}$, we have
\begin{eqnarray}
\frac{1}{n} \sum_{q=k}^{k+n-1} \left[ \log \frac{f_{\boldsymbol{\lambda}_{u}}(X_{1, q}, \ldots, X_{L, q})}{g_{\boldsymbol{\theta}_{0}}(X_{1, q}, \ldots, X_{L, q})} + |\log(1-\rho)| \right] \rightarrow d(\boldsymbol{\theta}_{0}, \boldsymbol{\lambda}_{u}), \quad  P^{(\boldsymbol{\theta}_{0}, \boldsymbol{\lambda}_{u})}_{k}\text{-almost surely}.
\end{eqnarray}
As a result, we have $\forall \delta \in (0, 1)$, $\exists n(\delta)$ such that $n \geq n(\delta)$
\begin{eqnarray}
P_{k}^{(\boldsymbol{\theta}_{0}, \boldsymbol{\lambda}_{u})}\left( \frac{1}{n} \sum_{q=k}^{k+n-1} \left[ \log \frac{f_{\boldsymbol{\lambda}_{u}}(X_{1, q}, \ldots, X_{L, q})}{g_{\boldsymbol{\theta}_{0}}(X_{1, q}, \ldots, X_{L, q})} + |\log(1-\rho)| \right] \leq d(\boldsymbol{\theta}_{0}, \boldsymbol{\lambda}_{u}) - \delta \right) \leq \delta. \label{eq:covergent}
\end{eqnarray}
In the following, we set
$$m_{c} = \lceil (1-\delta)^{-1} d^{-1}(\boldsymbol{\theta}_{0}, \boldsymbol{\lambda}_{u})c \rceil, $$
in which $c := \log B_{u} = |\log \alpha| + |\log \rho| + \sum_{l=1}^{L}\log I_{l}$, and $\delta>0$ is a small constant such that
\begin{eqnarray}
\liminf_{\alpha \rightarrow 0} \frac{m_{\alpha}}{m_{c}} \geq 1  \label{eq:ma_mc}
\end{eqnarray}
for any $\boldsymbol{\lambda}_{u} \in \mathbf{\Lambda}$.
Recall that
\begin{eqnarray}
\tilde{\tau}_{C}^{(u)} = \inf\{ n\geq 1: \tilde{C}_{n}^{(u)} \geq e^{c} \}. \no
\end{eqnarray}
Since $\tilde{\tau}_{C} = \min_{\boldsymbol{\lambda}_{u} \in \mathbf{\Lambda}} \tilde{\tau}_{C}^{(u)}$, then it is sufficient to present an upper bound for the detection delay of $\tilde{\tau}_{C}^{(u)}$. To this end, we use a technique that is similar to the one adopted in \cite{Lai:TIT:98}. Let $\boldsymbol{\lambda}_{u}$ be the true post-change parameter. Then, for any given $\delta>0$ such that \eqref{eq:ma_mc} is satisfied, when $\alpha$ is small enough, we have 
\begin{eqnarray}
&&P_{k}^{(\boldsymbol{\theta}_{0}, \boldsymbol{\lambda}_{u})}\left(\tilde{\tau}_{C}^{(u)} - k > rm_{c} \big| \tilde{\tau}_{C}^{(u)} \geq k \right) =  P_{k}^{(\boldsymbol{\theta}_{0}, \boldsymbol{\lambda}_{u})}\left(\tilde{\tau}_{C}^{(u)} > k+rm_{c} \big| \tilde{\tau}_{C}^{(u)} \geq k \right)\no\\
&&= P_{k}^{(\boldsymbol{\theta}_{0}, \boldsymbol{\lambda}_{u})}\left( \log \tilde{C}_{n}^{(u)} < c \text{ for all } n\leq k+rm_{c} \Big| \tilde{\tau}_{C}^{(u)} \geq k \right) \no\\
&&= P_{k}^{(\boldsymbol{\theta}_{0}, \boldsymbol{\lambda}_{u})}\left( \max_{n-m_{\alpha} \leq s \leq n} \sum_{q=s}^{n}\left[ \log \frac{f_{\boldsymbol{\lambda}_{u}}(X_{1, q}, \ldots, X_{L, q})}{g_{\boldsymbol{\theta}_{0}}(X_{1, q}, \ldots, X_{L, q})} + |\log(1-\rho)| \right] < c \right. \no\\
&&\hspace{30mm} \text{ for all } n\leq k+rm_{c} \Big| \tilde{\tau}_{C}^{(u)} \geq k \Bigg) \no\\
&&\leq P_{k}^{(\boldsymbol{\theta}_{0}, \boldsymbol{\lambda}_{u})}\left( \max_{n-m_{\alpha} \leq s \leq n} \sum_{q=s}^{n}\left[ \log \frac{f_{\boldsymbol{\lambda}_{u}}(X_{1, q}, \ldots, X_{L, q})}{g_{\boldsymbol{\theta}_{0}}(X_{1, q}, \ldots, X_{L, q})} + |\log(1-\rho)| \right] < c \right. \no\\
&&\hspace{30mm} \text{ for all } n = jm_{c}+k, j=1, 2, \ldots, r \Big| \tilde{\tau}_{C}^{(u)} \geq k \Bigg) \no\\
&&\overset{(a)}{\leq} P_{k}^{(\boldsymbol{\theta}_{0}, \boldsymbol{\lambda}_{u})}\left( \sum_{q=n-m_{c}+1}^{n}\log \frac{f_{\boldsymbol{\lambda}_{u}}(X_{1, q}, \ldots, X_{L, q})}{g_{\boldsymbol{\theta}_{0}}(X_{1, q}, \ldots, X_{L, q})} + m_{c}|\log(1-\rho)| < c \right. \no\\
&&\hspace{30mm} \text{ for all } n = jm_{c}+k, j=1, 2, \ldots, r \Big| \tilde{\tau}_{C}^{(u)} \geq k \Bigg) \no\\
&&= P_{k}^{(\boldsymbol{\theta}_{0}, \boldsymbol{\lambda}_{u})}\left(  \sum_{q=(j-1)m_{c}+k+1}^{jm_{c}+k}\log \frac{f_{\boldsymbol{\lambda}_{u}}(X_{1, q}, \ldots, X_{L, q})}{g_{\boldsymbol{\theta}_{0}}(X_{1, q}, \ldots, X_{L, q})} + m_{c}|\log(1-\rho)| < c \right. \no\\
&&\hspace{30mm} \text{ for all } j=1, 2, \ldots, r \Big| \tilde{\tau}_{C}^{(u)} \geq k \Bigg) \no\\
&&\overset{(b)}{=} \left[ P_{k}^{(\boldsymbol{\theta}_{0}, \boldsymbol{\lambda}_{u})}\left(  \sum_{q=k+1}^{m_{c}+k}\log \frac{f_{\boldsymbol{\lambda}_{u}}(X_{1, q}, \ldots, X_{L, q})}{g_{\boldsymbol{\theta}_{0}}(X_{1, q}, \ldots, X_{L, q})} + m_{c}|\log(1-\rho)| < c \right)\right]^{r} \no\\
&&\overset{(c)}\leq \left[ P_{k}^{(\boldsymbol{\theta}_{0}, \boldsymbol{\lambda}_{u})}\left( m_{c}^{-1} \sum_{q=k+1}^{m_{c}+k}\log \frac{f_{\boldsymbol{\lambda}_{u}}(X_{1, q}, \ldots, X_{L, q})}{g_{\boldsymbol{\theta}_{0}}(X_{1, q}, \ldots, X_{L, q})} + |\log(1-\rho)| < (1-\delta)d(\boldsymbol{\theta}_{0}, \boldsymbol{\lambda}_{u}) \right) \right]^{r} \no\\
&&\overset{(d)}\leq \delta^{r}, \label{eq:delay_r}
\end{eqnarray}
in which (a) is because of \eqref{eq:ma_mc} for $\alpha$ small enough; (b) is because $\{(X_{1, q}, \ldots, X_{L, q}), q=(j-1)m_{c}+k+1, \ldots, jm_{c}+k\}$ are i.i.d. over $j$ and are independent of event $\left\{\tilde{\tau}_{C}^{(u)} \geq k\right\}$; (c) is because $m_{c}$ is the smallest integer $\geq (1-\delta)^{-1}d(\boldsymbol{\theta}_{0}, \boldsymbol{\lambda}_{u})^{-1}c$; (d) is an application of \eqref{eq:covergent}, noting that for a given $\delta > 0$, $m_{c} \sim |\log \alpha| > n(\delta)$ when $\alpha$ is small enough.
With \eqref{eq:delay_r}, we have
\begin{eqnarray}
\mathbb{E}_{k}^{(\boldsymbol{\theta}_{0}, \boldsymbol{\lambda}_{u})}\left[ m_{c}^{-1}(\tilde{\tau}_{C}^{(u)} - k) \Big| \tilde{\tau}_{C}^{(u)} \geq k \right] \leq \sum_{r=0}^{\infty} \delta^{r} = (1-\delta)^{-1}.
\end{eqnarray}
As a result, we have
\begin{eqnarray}
\mathbb{E}_{\pi}^{(\boldsymbol{\theta}_{0}, \boldsymbol{\lambda}_{u})}\left[(\tilde{\tau}_{C} - t)^{+}\right] &\leq& \mathbb{E}_{\pi}^{(\boldsymbol{\theta}_{0}, \boldsymbol{\lambda}_{u})}\left[(\tilde{\tau}_{C}^{(u)} - t)^{+}\right] \no\\
&=& \sum_{k=1}^{\infty}\rho(1-\rho)^{k-1} \mathbb{E}_{k}^{(\boldsymbol{\theta}_{0}, \boldsymbol{\lambda}_{u})}\left[ (\tilde{\tau}_{C}^{(u)} - t)^{+} \right] \no\\
&=& \sum_{k=1}^{\infty}\rho(1-\rho)^{k-1} P_{k}^{(\boldsymbol{\theta}_{0}, \boldsymbol{\lambda}_{u})}\left(\tilde{\tau}_{C}^{(u)} \geq t \right) \mathbb{E}_{k}^{(\boldsymbol{\theta}_{0}, \boldsymbol{\lambda}_{u})}\left[\tilde{\tau}_{C}^{(u)} - k \big| \tilde{\tau}_{C}^{(u)} \geq k \right] \no\\
&\leq& \frac{m_{c}}{1-\delta} = \frac{c}{(1-\delta)^{2}d(\boldsymbol{\theta}_{0}, \boldsymbol{\lambda}_{u})}(1+o(1)).
\end{eqnarray}
Then the conclusion is achieved since the above equation holds for any given small $\delta > 0$.
\section{Numerical Simulation} \label{sec:sim}
In this section, we provide some numerical examples to illustrate the theoretical results obtained in this paper. 


In the first simulation, we illustrate the performance of the M-SR procedure and the modified M-SR procedure for Formulation \eqref{eq:P1}. In this simulation, we set $\rho = 0.01$, and we assume that the pre-change distribution $g_{\theta_0}$ is $\mathcal{N}(0, 1)$ and the post-change distribution $f_{\lambda}$ is $\mathcal{N}(\lambda, 1)$, where $\lambda$ takes value in a close interval $\Lambda = [0.4, 2.8]$. We set the true post-change parameter is $\lambda = 1$, but this value is unknown to the observer. To implement the proposed algorithms, the observer needs to select a discrete set $\Lambda_{D}$. In the simulation, we select two different discrete sets, namely $\Lambda_{D}^{1} = \{0.4, 1.6, 2.8\}$ and $\Lambda_{D}^{2} = \{0.4, 1, 1.6, 2.2, 2.8\}$, and compare the performances of the proposed algorithms on these two sets. The simulation result is shown in Figure \ref{fig:P1_Sigma1}. The black dot-dash line is the theoretical asymptotic lower bound calculated by \eqref{eq:LB_P1}. The red dash line with stars and the red solid line with squares are the performance of the modified M-SR procedure and that of the M-SR procedure w.r.t. $\Lambda_{D}^{1}$, respectively. The blue dash line with diamonds and the blue solid line with circles are the performance of the M-SR procedure and that of the M-SR procedure w.r.t. $\Lambda_{D}^{2}$, respectively. For both cases, the detection delay of the modified M-SR procedure is larger than that of the M-SR procedure. This is because $R_{n}^{(\theta_{0}, \lambda_{i})}$ is larger than $C_{n}^{(\theta_{0}, \lambda_{i})}$ as indicated by \eqref{eq:stat_R} and \eqref{eq:stat_C}. However, the gap between the M-SR procedure and the modified M-SR procedure are relatively small w.r.t. the detection delay as $\alpha \rightarrow 0$. This verifies the result in Theorem \ref{thm:general_opt} that the M-SR procedure and the modified M-SR procedure have the same asymptotic performance. Furthermore, we note that the lines in blue are almost parallel to the asymptotic lower bound but the lines in red are not. This indicates that the algorithms based on $\Lambda_{D}^{2}$ are asymptotic optimal since the constant differences between the blue lines and the black dash line are negligible when the detection delay goes to infinity. For the same reason, the algorithms based on $\Lambda_{D}^{1}$ suffer performance loss. This phenomenon is also indicated by the result in Theorem \ref{thm:general_opt}. Since one of the candidate points in $\Lambda_{D}^{2}$ is the true post-change parameter, we have $\min_{\lambda_{i} \in \Lambda_{D}^{2}}D(f_{\lambda}, f_{\lambda_{i}})=0$, hence blue lines are asymptotically optimal.

\begin{figure}[thb]
\centering
\includegraphics[width=0.5 \textwidth]{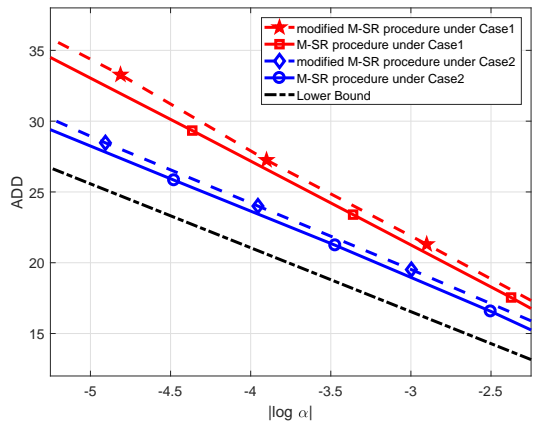}
\caption{The performance of M-SR procedure and modified M-SR procedure under different settings}
\label{fig:P1_Sigma1}
\end{figure}

In the second simulation, we examine the asymptotic optimality of the proposed window-based modified M-SR procedure for the multi-source monitoring problem considered in Section \ref{sec:ext}. In this simulation, we consider three mutually independent sources. For all of these three sources, the pre-change distribution is $\mathcal{N}(0, 1)$ and the post-change distribution is $\mathcal{N}(0, \lambda_{i}^2)$ for $i = 1, 2, 3$. In the simulation, we set $\mathbf{\Lambda} = \Lambda_{1} \times \Lambda_{2} \times \Lambda_{3}$ with $\Lambda_{i} = \{1.5, 1.6, 1.7, 2, 2.1, 2.2, 2.3\}$ for $i = 1, 2, 3$. However, the true post-change parameter for each source is different: particularly, $\lambda_{1}=1.7$, $\lambda_{2}=2$ and $\lambda_{3}=2.2$. In addition, we set $\rho = 0.01$ and the window length $m_{\alpha} = 200$. The performance of the proposed window-based modified M-SR procedure is presented in Figure \ref{fig:multisensor_case}. The black dot line is the theoretical asymptotic lower bound calculated by \eqref{eq:multi_LB} and the blue line with squares is the performances of the proposed algorithm. We can see that the blue line is parallel to the theoretical asymptotic lower bound, hence the proposed modified M-SR procedure is asymptotically optimal.
\begin{figure}[thb]
\centering
\includegraphics[width=0.52 \textwidth]{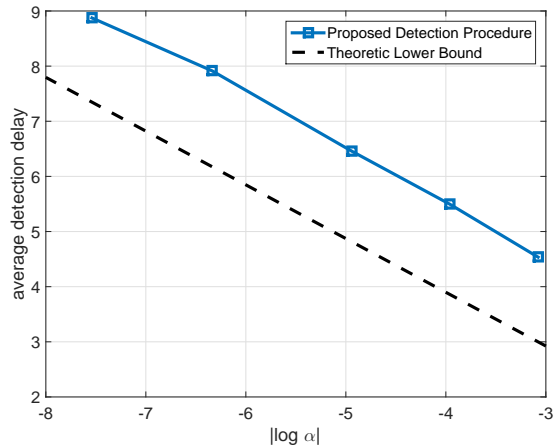}
\caption{The performance of window based modified M-SR procedure for three mutually independent sequences}
\label{fig:multisensor_case}
\end{figure}

%

\section{Conclusion} \label{sec:con}
In this paper, we have considered the Bayesian QCD problem with unknown post-change parameters. In this case, we have proposed two low complexity multi-chart detection procedures, namely the M-SR procedure and the modified M-SR procedure, and have shown these two multi-chart detection procedures are asymptotically optimal when $\Lambda$, the feasible set of the post-change parameter, is finite and asymptotically $\epsilon-$optimal when $\Lambda$ is an interval. We have also considered the multi-source monitoring problem with unknown post-change parameters. In this case, we have proposed a window-based modified M-SR detection procedure and have shown its asymptotic optimality. Both the computational complexity and storage requirement of the proposed algorithm have been shown to be on the order of $O(m_{\alpha}\sum_{l=1}^{L} I_{l})$.

\bibliographystyle{ieeetr}{}
\bibliography{macros,highdimension,detection,application}

\end{document}